\documentclass[alpha-refs]{wiley-article}
\usepackage{graphicx} 
\usepackage{xcolor}
\usepackage{comment}
\usepackage{csquotes}
\usepackage{booktabs}
\usepackage{multirow}
\usepackage{bbm}
\usepackage{adjustbox}
\usepackage{longtable}
\usepackage{array}
\usepackage{rotating}
\usepackage{pdflscape}
\usepackage{tikz}
\usetikzlibrary{positioning,arrows.meta,calc}

\usepackage{natbib}

\usepackage{caption}
\makeatletter
\def\@opargbegintheorem#1#2#3{%
  \trivlist
  \item[\hskip\labelsep\textbf{#1\ #2\ (#3).}]%
  \itshape
}

\newtheorem{mydefinition}{Definition}
\newtheorem{myassumption}{Assumption} 
\newtheorem{myremark}{Remark}
\newtheorem{mylemma}{Lemma}
\newtheorem{myproposition}{Proposition}
\newtheorem{mycorollary}{Corollary}
\newtheorem{mytheorem}{Theorem}

\newenvironment{myproof}[1][Proof]{%
  \par\noindent\textbf{#1. }\ignorespaces
}{%
  \hfill$\blacksquare$\par
}

\usepackage{siunitx}

\papertype{Original Article}
\paperfield{Journal Section}

\title{Market-implied time to transition to a low-carbon economy: a stochastic modelling and inference framework}

\author[1]{Lorenzo Mercuri}
\author[2]{Andrea Perchiazzo}
\author[3]{Edit Rroji}
\author[4]{Ilaria Stefani}

\affil[1]{Department of Economics, Management and Quantitative Methods, University of Milan, Italy.}

\affil[2]{Department of Economics and Business Studies, University of Piemonte Orientale, Italy.}

\affil[3]{Department of Statistics and Quantitative Methods, University of Milano-Bicocca, Italy.}

\affil[4]{Institute of Insurance Science, University of Ulm, Germany.}

\corraddress{Lorenzo Mercuri, Department of Economics and Quantitative Methods, University of Milan, Italy}
\corremail{lorenzo.mercuri@unimi.it}

\fundinginfo{This work was supported by JST CREST Grant Number JPMJCR2115, Japan
and by the European Union - NextGeneration EU PRIN2022 project ``The effects of climate change in the evaluation of financial instruments'' financed by the ‘Ministero dell’Università e della Ricerca’
with grant number 20225PC98R, CUP Codes: H53D23002200006 and G53D25001960006. Andrea Perchiazzo gratefully acknowledges the University of Piemonte Orientale (UPO) for the support provided through the Starting Kit 2025.}

\runningauthor{Mercuri et al.}

\begin{document}

\begin{frontmatter}
\maketitle

\begin{abstract}
This paper introduces a new market-implied object, \textit{Time to Transition} (TtT), extracted from the difference between two selected nodes of the greenium term structure. TtT is defined as the latent waiting time until this cross-maturity greenium difference vanishes, meaning that the greenium becomes equal across the two selected maturities. We develop an inference theory for this object. To model TtT, we introduce two tractable stochastic frameworks: the \textit{Regulatory Deadline-Constrained Model}, in which the transition date is fixed, and a switching extension, in which alternative transition dates capture heterogeneous perceived deadlines across economic agents. The paper combines two layers of analysis. On a fixed daily grid, a deadline-constrained diffusion provides a tractable benchmark through an exact Gaussian bridge likelihood, while the switching extension preserves tractability through regime-specific bridge densities and filtering recursions. Under a \textit{fixed-horizon infill scheme}, the same framework yields a structural identification result for the regime-wise diffusion parameters, with full or partial consistency depending on the observed region. The paper therefore contributes both a new inferential object, market-implied transition timing based on cross-maturity differences in the greenium term structure, and a two-layer inference framework: finite-sample filtering provides an operational monitoring tool, while fixed-horizon infill asymptotics specify when the regime-wise diffusion parameters carrying information about competing transition dates can be consistently estimated.
% Please include a maximum of seven keywords
\end{abstract}
\end{frontmatter}

\section{Introduction}

The European Union (EU) has made the transition to a low- or zero-carbon economy a central objective in the pursuit of a climate-neutral society, setting explicit targets, policies, and timelines. According to the European Climate Law,\footnote{Regulation (EU) 2021/1119 of the European Parliament and of the Council of 30 June 2021 establishing the framework for achieving climate neutrality and amending Regulations (EC) No 401/2009 and (EU) 2018/1999 (``European Climate Law''). Available at: https://eur-lex.europa.eu/eli/reg/2021/1119/oj/eng. For further details on EU climate action see also https://climate.ec.europa.eu/} the EU aims to reduce greenhouse gas (GHG) emissions by at least $55\%$ by 2030, relative to 1990 levels, and to reach climate neutrality by 2050. These policy targets provide a clear institutional roadmap; however, their economic relevance depends on whether the corresponding timelines are regarded as credible by market participants. The way economic agents perceive such deadlines may affect investment, consumption, policy support, and, more generally, the pricing of transition risk.

Financial prices often reflect not only beliefs about future states of the world, but also beliefs about the timing with which those states may materialize. Recovering such timing information from market data is, however, intrinsically difficult: the relevant date is latent, the observable signal is noisy, and finite-sample evidence need not carry a structural interpretation. This paper studies that problem in a specific but analytically tractable environment and develops an inference framework for a market-implied transition time.

Our observed quantity is the difference between two selected nodes of the greenium term structure extracted from pairs of green and conventional bond prices. We interpret this difference as a market-based signal about the timing of transition to a low-carbon economy. The object is forward-looking: rather than measuring the level of the green premium at a given date, it uses cross-maturity information to infer when investors collectively price the transition as having effectively materialized.  In a similar spirit to the VIX, which extracts forward-looking information from option prices, our goal is to extract from bond prices a market-implied measure of the timing of the low-carbon transition. More precisely, we introduce a market-implied \emph{Time to Transition} (TtT), defined as the waiting time until the greenium term structure becomes flat across maturities. In line with \cite{damico}, we argue that investors are more willing to sacrifice financial returns in the short term than in the long term. This behaviour explains the inverted greenium term structure observed in the German twin bond market. Once the transition to a low-carbon economy is accomplished, we expect investors to believe that their task is complete, as the zero-carbon target will be linked to technological innovation. This expectation motivates the hypothesis of flat greenium term structure in the long run.
The empirical environment is the German sovereign twin-bond market,\footnote{A twin bond refers to a bond issuance structure in which two bonds share identical financial characteristics (e.g., issuer, maturity, coupon) but different labels or funding purposes. For further details, see ``Twin Bond Concept'' of the ``Green Federal Bonds'' section available at https://www.deutsche-finanzagentur.de/en/federal-securities.} observed at daily frequency from 8 September 2021 to 17 January 2025. Its relevance is methodological before being substantive. Since the green and conventional securities differ only in their label, this market provides an unusually clean setting in which the informational content of the green label can be isolated without relying on matching procedures. In this sense, the application is best viewed as a disciplined laboratory for inference on transition timing, rather than as a conventional study of green-bond premia.

The contribution of the paper is not to revisit the existence or the average magnitude of the greenium. That question has already received substantial attention in the literature; see, among others, \cite{ehlers2017green, hachenberg2018green, gianfrate2019green, larcker2020s, baker2018financing, karpf2018changing, Zerbib2019, wang2020market, flammer2021corporate, huang2023rethinking}.
Our question is different: can the difference between two nodes of the greenium term structure be used to extract information about perceived transition timing,
and under which observation schemes does such inference admit a structural
interpretation? This shift in focus moves the analysis from premium measurement
to stochastic inference on latent temporal labels.

% Our question is different: can the difference of two nodes in the greenium term structure be used to infer perceived transition timing, and under which observation schemes does such inference admit a structural interpretation? This shift in focus moves the analysis from premium measurement to stochastic inference on latent temporal labels.

To address this question, we first introduce a benchmark \emph{Regulatory Deadline-Constrained Model} (RDCM). The RDCM is a mean-reverting diffusion with terminal condition, designed so that the time remaining to a candidate deadline affects the pre-terminal dynamics. Economically, it provides a disciplined benchmark for a single perceived transition date. Statistically, it yields an exact Gaussian bridge likelihood, which makes the model tractable despite the terminal constraint. At the same time, this benchmark is intentionally restrictive: it cannot accommodate heterogeneity in perceived transition timing and, when the sample lies entirely before the structural time point $\tau$, it cannot fully identify the components of the model that become visible only closer to the candidate deadline.

This limitation motivates our main extension, the \emph{Switching Regulatory Deadline-Constrained Model} (SRDCM), in which alternative perceived transition dates coexist as latent states. On a fixed observation grid, the SRDCM is used as a filtering device: it reallocates posterior probability mass across competing deadline scenarios and provides an empirical monitoring tool for changes in perceived transition timing.

The theoretical part of the paper addresses a separate identification question. Under a fixed-horizon infill asymptotic scheme, we study when the regime-wise diffusion parameters associated with the competing deadlines can be consistently estimated. The analysis relies on the consistency theory of \(M_n\)-contrast functions; see, among others, \cite{van1998asymptotic, IbragimovHasminskii1981StatisticalEstimation, Kutoyants2004StatisticalInferenceStochasticProcesses, Yoshida2022QLA, cheng2026quasi}. The resulting full or partial consistency depends on which part of the diffusion coefficient is visible over the observation window.

The remainder of the paper is organized as follows. Section \ref{term} focuses on the greenium term-structure and the associated notion of market-implied Time to Transition.  Section \ref{newmodel} develops the RDCM and the SRDCM, discusses their probabilistic structure, and derives the corresponding likelihood and filtering machinery. Section \ref{empanalysis} presents the empirical filtering analysis. Section \ref{sec:highfreq} studies the consistency property of the diffusion parameters under fixed-horizon infill asymptotics. Section \ref{concl} concludes.

\section{Greenium term structure}\label{term}
In this section, we introduce a node-difference measure of the greenium term structure, constructed as the difference between two selected maturity nodes using pairs of twin bonds. Building on this measure, we define a transition point at which the difference between the two nodes vanishes, indicating that the greenium is equal across the selected maturities, and we introduce the associated concept of Time to Transition. Using these theoretical tools, we present empirical evidence illustrating the behaviour of greenium node differences and the related transition features observed in the data.

\subsection{Greenium term structure and market implied Time to Transition}\label{defi}

For each observation date $t$ and maturity $\bar{T}$, we consider market quotations of a  pair of zero-coupon bonds (for example twin bonds). We assume these bonds to have the same issuer (sovereign or private), to be denominated in the same currency, to share the same level of subordination, and not to have embedded optionalities. More formally, $D_{g}\left(t,\bar{T}\right)$ and $D_{b}\left(t,\bar{T}\right)$ denote respectively the price of green and brown zero-coupon bonds with maturity $\bar{T}$. The \textsl{greenium} $g_r(t, \bar{T})$ on day $t$ for maturity $\bar{T}$, seen as the difference of the two yield to maturities, can be expressed as: 
\begin{equation}
g_r\left(t,\bar{T}\right) :=-\frac{1}{\bar{T}-t}\left\{\log\left(D_{b}\left(t,\bar{T}\right)\right)-\log\left(D_{g}\left(t,\bar{T}\right)\right)\right\}.
\label{eq:greenium}
\end{equation} 
\indent In order to extract some information from the greenium term structure, we consider a quantity $X_t$  defined as the difference in greenium between two selected maturities of twin bonds. 
\begin{mydefinition}
Let $\mathcal{T}\left( t\right):= \left\{T_i^{\left(t\right)}\right\}_{i=1, \dots, n}$ be the set of available maturities for which (observable) twin bond exists
at time $t$. Choosing two reference maturities $T_{\text{short}} \left(t\right) \in \mathcal{T}\left( t\right)$ and $T_{\text{long}} \left(t\right) \in \mathcal{T}\left( t\right)$, satisfying
\begin{equation*}
    T_{\text{long}} \left(t\right) - T_{\text{short}} \left(t\right) = h, \quad \forall t:t <  T_{\text{short}}\left(t\right),
\end{equation*}
then $X_t$ is defined as
\begin{equation}
    X_t:= \Delta g_r\left(t, T^{\text{short}}_t , T^{\text{long}}_t \right) = g_r\left(t,T_{\text{short}} \left(t\right)\right) - g_r\left(t,T_{\text{long}} \left(t\right)\right)
    \label{xt}.
\end{equation}
\end{mydefinition}
We assume that the process $X_t$ is not observable from its starting point, but only from the dates that a pair of twin bonds is quoted.  We also assume the process $X_t$ to become constant at a future time when the transition occurs.
We now define the Transition Point and the Time to Transition.
% as the earliest future time at which the difference in greenium vanishes and the greenium term structure remains flat thereafter, that is:
\begin{mydefinition}\label{Def:TimetoTrans}
The transition point $T$ is the earliest future time at which the difference in greenium vanishes and the greenium term structure remains flat thereafter. That is:
\begin{equation}
T := \inf\left\{t \geq0: \left\{X_s\right\}_{s \geq t} =0 \right\}. 
\end{equation}
Accordingly, the \textit{Time to Transition}, denoted by $\text{TtT}\left(t,T\right)$, is:
\begin{equation}\label{eq_ttt}
    \text{TtT}\left(t, T\right):= T - t.
\end{equation}
\end{mydefinition}
Intuitively, the Time to Transition is the ``waiting'' time until the greenium term structure becomes flat across maturities.
%the Time to Transition coincides with the ``waiting'' time in order to observe a flat greenium term structure, 
Specifically, we admit a premium for the label green even if the transition occurs, but it should not be maturity-specific.  In line with \cite{damico}, we argue that investors are more willing to sacrifice financial returns in the short term than in the long term. This behaviour explains the inverted greenium term structure observed in the German twin bond market (see \citealt{mercuri2025} for instance).
Investors do indeed care about the environment, as they are willing to sacrifice financial returns in order to fund green projects, thereby effectively subsidising the green transition. However, they are influenced by real-world events, such as political elections, wars, and financing costs, which can make the transition less realistic within the scheduled time-frame. In particular, we assume that investors may decide that the target of reducing emissions by 55\% is unachievable in the scheduled timeline, meaning their financial sacrifice is no longer worthwhile and the greenium term structure becomes less steep. It is reasonable to think that investors may accept a lower remuneration only if they believe that the positive effects (both economical and social) of the transition would materialize in the near future. In our framework, we allow investors to be influenced by subsequent events; i.e., we could expect periods of acceleration or deceleration of the perceived timeline for the transition.\newline
Conversely, once the transition to a low-carbon economy is accomplished, we expect investors to believe that their task is complete, as the zero-carbon target will be linked to technological innovation. This expectation motivates the flat greenium term structure in the long run.

 \subsection{Empirical observations from German twin bonds}\label{emp}

Building upon the definition of $X_t$ in \eqref{xt}, which we recall is the difference between two nodes in the greenium term structure, we now proceed to investigate its behaviour through an empirical analysis. We consider the German sovereign bond market motivated by the fact that 
% two considerations: i) Germany exhibited relatively high greenhouse gas emissions (GHG) during the period 2019-2023 compared to other EU Member States ; ii) 
Germany is one of the few EU countries that issues twin bonds, which allows us to make an analysis in a distortion-free setting (as the bonds differ only in their label).

% \begin{figure}[!htbp]
%     \centering
%     \includegraphics[width=0.7\textwidth]{img_paper/GHG.eps}
%     \caption{GHG emissions in Germany during 2019-2023 compared with a sample of EU peers}\label{fig:emission_germany}
% \end{figure}

To this end, we select two pairs of German twin bonds (highlighted in bold in Table~\ref{tab:twin_bonds}) maturing in 2025 (ISINs DE0001141828 and DE0001030716) and 2030 (ISINs DE0001102507 and DE0001030708) thereby focusing on a short- to medium-term part of the green bond term structure. These two reference maturities are included because, as observed in \cite{mercuri2025},  the greenium term structure is significantly steeper at earlier maturity nodes. The selected set comprises only zero-coupon bonds with no embedded optionalities. Data are downloaded from the Bloomberg terminal.
\begin{table}[!htbp]
\centering
% \scalebox{0.78}{
\scalebox{0.65}{
\begin{tabular}{lllcc}
\hline
\textbf{ISIN}                                   & \textbf{Maturity} & \textbf{Type}     & \textbf{First Coupon}       & \textbf{Coupon Rate} \\
\hline
\hline
 \textbf{DE0001141828}           &  $10/10/2025$ & \textbf{Brown} & -- &--  \\
 \textbf{DE0001030716}           &  \textbf{$10/10/2025$} &  \textbf{Green} & -- & -- \\
 DE0001141869           & $15/10/2027$ & Brown &$15/10/2023$ & $1.3$ \\
 DE0001030740           & $15/10/2027$ & Green & $15/10/2023$ & $1.3$   \\
 \textbf{DE0001102507}           &  \textbf{$15/08/2030$}  & \textbf{Brown}  & -- &--\\
 \textbf{DE0001030708}           &  \textbf{$15/08/2030$}&  \textbf{Green} & -- & --\\
 DE0001102564           & $15/08/2031$    & Brown& -- &-- \\
 DE0001030732          &$15/08/2031$  & Green & -- & --\\
 DE000BU2Z007          & $15/02/2033$& Brown& $15/02/2024$& $2.3$\\
 DE000BU3Z005         &$15/02/2033$& Green& $15/02/2024$&$2.3$  \\
 DE0001102481          & $15/08/2050$& Brown& --& --\\
 DE0001030724       &$15/08/2050$& Green& --&-- \\
 DE0001102614         & $15/08/2053$& Brown& $15/08/2023$& $1.8$\\
 DE0001030757        &$15/08/2053$& Green& $15/08/2023$&$1.8$  \\
\hline
\end{tabular}
}
\caption{Features of German twin bonds available in the market as of January 17, 2025.} \label{tab:twin_bonds}
\end{table}

% \subsubsection{Greenium term-structure slope and yield curve volatility}

As an exploratory investigation, we compare the time series of the cross-maturity greenium difference $X_t$ (blue line) from October 2021 to January 2025 with the 30-day rolling volatility of the German government yield curve across the 3-month, 1-year, 2-year, and 4 year maturity nodes (see Figure~\ref{fig:rolling_vola}). 

\begin{figure}[h!]
    \centering
    \includegraphics[width= 1\textwidth]{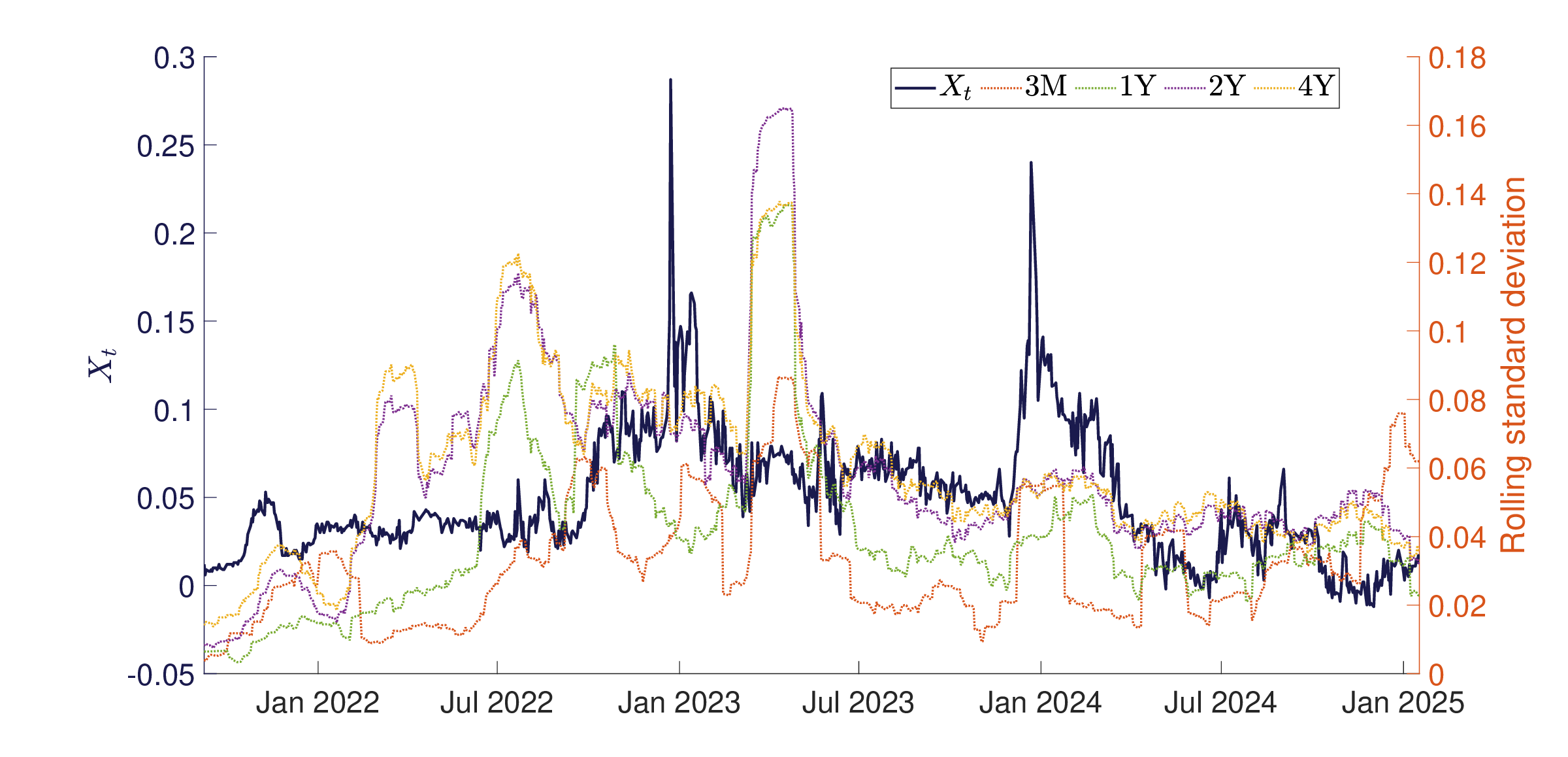}
    \caption{Time series of the cross-maturity greenium difference $X_t$ (blue line, left axis multiplied by 100) from $08/09/2021$ to $17/01/2025$ for the 2025 and  the 2030 maturity nodes compared with 30-day rolling (right axis) volatility of the German yield curve (3M, 1Y, 2Y, 4Y).}
    \label{fig:rolling_vola}
\end{figure}

The visual comparison shown in Figure~\ref{fig:rolling_vola} is used to explore whether the dynamics of $X_t$, and thus the transition expectation toward a low-carbon economy, are associated (or move in line) with conditions in the market interest rate environment. The figure displays that the difference in greenium $X_t$ generally fluctuates at low to moderate levels, with only two significant sharp spikes. After each spike, the series tends to revert toward its baseline, suggesting a mean-reverting behaviour. Regarding the two spikes in the time series of $X_t$, they occur around the introduction of two regulatory directives. The first spike is observed after Regulation (EU) 2022/2577 of the Council (22 December 2022), which establishes a framework to accelerate the development of renewable energy. The second spike appears following Regulation (EU) 2024/223 of the
Council (22 December 2023) that amends Regulation (EU) 2022/2577 with the same objective. The latter regulation, titled ``Accelerating the deployment of renewable energy'', introduces an emergency framework based on Article 122 of the \textit{Treaty on the Functioning of the European Union} for accelerating the diffusion of renewable energy in the short term. It took effect on 30 December 2022, initially valid for 18 months, but was later extended by Regulation (EU) 2024/223 following a review by the European Commission. 

By contrast, the regulatory milestones related to the EU Green Bond Standards for new bond issuances occurred over several key dates\footnote{The EU Green Bond Standards were developed through a gradual regulatory process. Key milestones include the publication of Technical Expert Group (TEG) Interim report and Final Reports in 2019, the European Commission's proposal in July 2019, and the adoption of Regulation (EU) 2023/2631 in October 2023. Such a Regulation entered into force in December 2023, with application date set for December 2024, and was completed by delegated acts in April 2025.} but do not appear to have a direct impact on the series of $X_t$ computed using German twin bonds. This absence of a direct impact of these new standards on government bonds may suggest that the ``green'' label for the traded German Bonds is broadly coherent with the new standards.

The financing of the transition to a low-carbon economy is closely linked to market interest rates. It is therefore reasonable to hypothesise that uncertainty in yield, for instance arising from frequent announcements by the European Central Bank during the post-pandemic inflation period, could affect how the timing of the transition is perceived. 

Although Figure~\ref{fig:rolling_vola} provides valuable insights, it does not allow us to identify the type of relationship (e.g., nonlinear or threshold type) between $X_t$ and interest rate uncertainty. It is important to note that the lack of clear linear co-movement does not necessarily imply the absence of asymmetric or regime-dependent effects. In such cases, yield volatility may influence transition expectations only beyond certain levels or during specific periods.

\section{Perceived time to transition}\label{newmodel}
In the transition toward a low-carbon economy, we expect the difference in greenium between pairs of twin bonds with different maturities to disappear or, at least, to decline in both level and volatility. This motivates the introduction of a model that imposes a terminal condition on the dynamics of the process representing differences across nodes of the greenium term structure. Empirically, $X_t$ exhibits a mean-reverting behaviour (see Figure~\ref{fig:rolling_vola}). Moreover, the dynamics of $X_t$ appear to change across low and high volatility periods, suggesting the presence of regime shifts. To capture these features while allowing for a transition occurring at a specific date, we consider two models. The first, referred to as the Regulatory Deadline-Constrained Model and discussed in Section~\ref{sec_RDCM}, is an extension of the classical \cite{vasicek} model where the drift and the diffusion coefficients  do not depend on \text{TtT} until a time instant that we name $\tau$ with $t<\tau<T$. Starting from $\tau$,
both drift and diffusion coefficients become dependent on \text{TtT}.
% we have a dependence of drift and diffusion w.r.t \text{TtT}. 
Typically, the instantaneous volatility and drift functions remain constant until this date, after which they decrease. The second model, presented in Section~\ref{SectioSRCDM}, incorporates this Regulatory Deadline-Constrained Model specification within a discrete-time regime-switching framework
where the perceived deadline for transitioning to a low-carbon economy defines the regime.

\subsection{Regulatory Deadline-Constrained Model}\label{sec_RDCM}

 We consider two functions $\phi\left(s,b\right):\mathbb{R}\times\mathbb{B}\rightarrow \mathbb{R}, \ \mathbb{B}\subseteq\mathbb{R}^{n_1}$ and $g\left(s,\theta\right):\mathbb{R}\times \Theta\rightarrow \left[0, +\infty\right), \ \Theta\subseteq\mathbb{R}^{n_2}$ with $n_1,n_2$ positive integers that satisfy the following assumption. 

\begin{myassumption}\label{rdcm_assumption}
For any $b \in \mathbb{B}$, $s\mapsto\phi\left(s,b\right)$ has  the following properties:  
\begin{enumerate}
    \item[(i)] it is continuous and integrable; 
    \item[(ii)] (non-negativity) $\phi\left(s,b\right)\geq 0$  $\forall s >0$ and $\phi\left(s,b\right)=0$ for $s\leq 0$.
\end{enumerate}
For any $\theta \in \Theta$, $s\mapsto g\left(s,\theta\right)$ satisfies the following requirements:
\begin{enumerate}
    \item[(i)] it is continuous and square-integrable; 
    \item[(ii)] (non-negativity) $g\left(s,\theta\right)> 0$  $\forall s >0$  and $g\left(s,\theta\right)=0$  $\forall s\leq 0$.

    \end{enumerate}    
\end{myassumption}
The function \(\phi(s,b)\) is used to describe the deterministic target level of the dynamics for $X_t$, whereas \(g(s,\theta)\) controls the magnitude of the random fluctuations around that target.

We make explicit that the dependence of the coefficient functions $\phi\left(\cdot,b\right)$ and $g\left(\cdot,\theta\right)$ on $TtT:=T-t$ (see Definition \ref{Def:TimetoTrans}), is activated only after a deterministic time $\tau<T$, which $\tau$ is named structural time point, fixed in advance. For any $b\in \mathbb{B}$ and $\theta \in \Theta$, we write: 
\[
b=(b^{-},b^{+}), \qquad \theta=(\theta^{-},\theta^{+}),
\]
where $b^{-}$ and $\theta^{-}$ are effective for all $t<T$, whereas $b^{+}$ and $\theta^{+}$ are effective only on the interval $[\tau,T)$.\footnote{In the practical specifications considered below, $b^{-}$ and $\theta^{-}$ determine the constant pre-$\tau$ level and remain effective also after $\tau$, while $b^{+}$ and $\theta^{+}$ govern the decay in the convergence phase.}

\begin{myassumption}
\label{ass:tau_activation}
For a fixed transition point $T>0$, there exists $\tau\in[0,T)$ and continuous functions
\[
f_{\tau^-}:[0,\tau]\to\mathbb{R}_+,
\qquad
f_{\tau^+}:[0,T-\tau]\to\mathbb{R}_+,
\]
such that
\[
s_{T,\tau}(t):=
f_{\tau^-}(t)\mathbf{1}_{\{t<\tau\}}+
f_{\tau^+}(T-t)\mathbf{1}_{\{\tau\le t\le T\}},
\]
with
\[
f_{\tau^-}(\tau)=f_{\tau^+}(T-\tau),
\qquad
f_{\tau^+}(0)=0.
\]
\end{myassumption}
Assumption~\ref{ass:tau_activation} is maintained throughout the paper unless otherwise stated.

\begin{myremark}
\label{rem:tau_notation}
Throughout the paper we write
\[
\phi(T-t,b) := \phi\bigl(s_{T,\tau}(t),b\bigr),
\qquad
g(T-t,\theta) := g\bigl(s_{T,\tau}(t),\theta\bigr).
\]
\end{myremark}

Under Assumptions \ref{rdcm_assumption} and \ref{ass:tau_activation}, we now introduce the Regulatory Deadline-Constrained Model (hereafter RDCM), which describes the evolution of the difference between two nodes in the greenium term structure. The model assumes that the transition occurs deterministically at a fixed regulatory deadline $T$ (for example, $T = 31/12/2030$ aligned with the setting of EU regulations and so a regulatory deadline), at which point the greenium difference across maturities vanishes. 

\begin{mydefinition}\label{RDCMdef}
Let $\left(\Omega,\mathcal{F},\mathbb{F},\mathbb{P}\right)$ be a filtered probability space with filtration $\mathbb{F} =\left\{\mathcal{F}_t\right\}_{t\geq 0}$ and let $\left\{W_t\right\}_{t\ge0}$ be a Wiener process adapted to $\mathbb{F}$. Let $T>0$ denote the regulatory deadline for the transition. Under Assumptions \ref{rdcm_assumption} and \ref{ass:tau_activation}, a stochastic process $\left\{X_t\right\}_{t\ge 0}$ is called a Regulatory Deadline-Constrained Model if it satisfies the stochastic differential equation
\begin{equation}
\mbox{d}X_t = a\left[\phi\left(T-t,b\right)-X_t\right]\mbox{d}t+g\left(T-t,\theta\right) \mbox{d}W_t,
\label{RDCM}
\end{equation}
subject to the terminal condition $X_T = 0$, where $a\in[0+\infty)\subset\left(0,+\infty\right)$, $b \in \mathbb{B}\subset\mathbb{R}^{n_1}$ and $\theta\in \Theta\subset\mathbb{R}^{n_2}$. 
% $\mathcal{A}$, $\mathbb{B}$ and $\Theta$ are compact sets.
\end{mydefinition}
Property (ii) in Assumption \ref{rdcm_assumption} for both coefficient functions ensures that $X_{t}=0$ for all $t\geq T$.

The economic interpretation of this result is that, once the transition to a low-carbon economy is complete, the greenium term structure becomes flat and is no longer affected by bond maturities. At this stage, the model implies a structural convergence where the greenium ceases to compensate for transition risk, which has been neutralized by the adoption of zero-impact technology, and instead serves as a constant premium for the intrinsic value of the green label. Therefore, in the post-transition regime (\(t\geq T\)), the greenium reflects purely idiosyncratic investor preferences for certified green assets, manifesting as a parallel shift relative to brown bonds rather than a maturity-dependent spread. 

In the empirical analysis (Section~\ref{empanalysis}), a convenient specification is the following:
\begin{align}\label{eq:phiandg}
\phi(T-t,b)
&=
b^{-}\min\left\{1,\left[\max\left(\frac{T-t}{T-\tau},0\right)\right]^{b^{+}}\right\},
\nonumber\\
g(T-t,\theta)
&=
\theta^{-}\min\left\{1,\left[\max\left(\frac{T-t}{T-\tau},0\right)\right]^{\theta^{+}}\right\}.
\end{align}
In this class of specifications, $\phi$ and $g$ are constant on $[0,\tau]$, then decay on $(\tau,T)$ as functions of $TtT$, and vanish at $t=T$. Accordingly, for $t\le\tau$ the dynamics reduce to a Vasicek-type regime, whereas for $\tau<t<T$ the process enters a final convergence phase toward the regulatory deadline. Note that $\phi$ and $g$ are continuous at $t=\tau$, although they need not be differentiable there.

\begin{myremark}
Treating $\tau$ as an additional model parameter is left for future research. This would substantially complicate the estimation procedure, since it introduces a change-point effect in the likelihood function; see \citealt{IACUS20121068} and the references therein.
\end{myremark}
Figure \ref{simulated} shows a simulated trajectory of the RDCM model in which both coefficient functions take the form \eqref{eq:phiandg}. Similar functions have been applied in \cite{lyashenko} for pricing  backward-looking caplets.
\begin{figure}[!htbp]
    \centering
\includegraphics[width=0.7\textwidth]{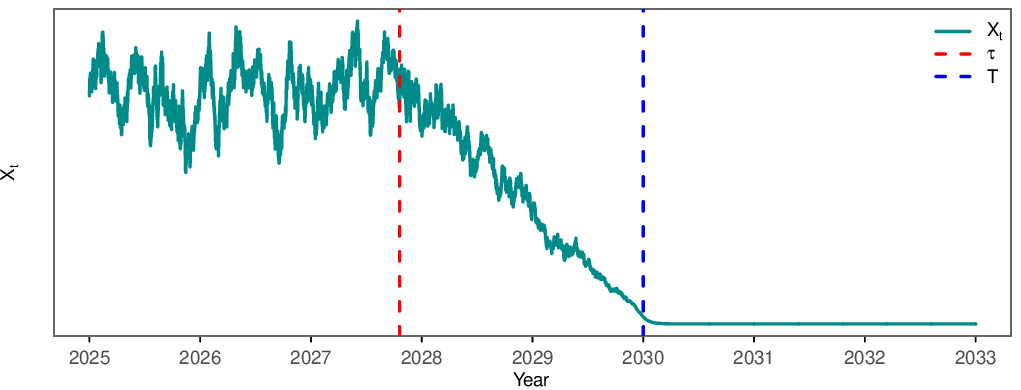}
    \caption{Possible trajectory of the process $X_t$ in the Regulatory Deadline-Constrained Model.\label{simulated}}
    % \label{simulated}
\end{figure}

\noindent The terminal condition $X_T=0$ suggests that classical estimation methods for mean-reverting models cannot be applied in our framework. In the following, we show how to compute the log-likelihood function in the RDCM framework where we exploit the fact that the conditional transition density given the final value is still Gaussian.

\subsubsection{Calibration of RDCM model}\label{Estim_RDCM_SUB}
We recognize that \eqref{RDCM} is a linear SDE with time-varying coefficients. Hence, for any $t>t_0$, the conditional law of $X_t$ given $\{X_{t_0}=x_{t_0}\}$ is Gaussian, namely
\[
X_t \mid X_{t_0}=x_{t_0} \sim \mathcal{N}\bigl(m_{t\mid t_0},v_{t\mid t_0}\bigr),
\]
with
\begin{equation}
m_{t\mid t_0}
:= \mathbb{E}\bigl[X_t\mid X_{t_0}=x_{t_0};a,b\bigr]
= e^{-a(t-t_0)}x_{t_0}
+a\int_{t_0}^{t} e^{-a(t-s)}\phi(T-s,b)\,\mathrm ds,
\label{eq:mean}
\end{equation}
\begin{equation}
v_{t\mid t_0}
= \mathbb{V}\bigl[X_t\mid X_{t_0};a,\theta\bigr]
= \int_{t_0}^{t} e^{-2a(t-s)} g^2(T-s,\theta)\,\mathrm ds.
\label{eq:var}
\end{equation}
Here $m_{t\mid t_0}$ depends on the realized initial value $x_{t_0}$, whereas $v_{t\mid t_0}$ does not.

\noindent For any $t_0<t\le T$, the pair $(X_T,X_t)$, conditional on $\{X_{t_0}=x_{t_0}\}$, is jointly Gaussian. Its conditional covariance is
\begin{equation}
\mathrm{Cov}_{t_0}(X_T,X_t)
:= \mathbb{E}\!\left[\bigl(X_T-m_{T\mid t_0}\bigr)\bigl(X_t-m_{t\mid t_0}\bigr)\mid X_{t_0}=x_{t_0}\right]
= e^{-a(T-t)}v_{t\mid t_0}.
\end{equation}
Therefore, by Gaussian conditioning, the transition law of the RDCM in Definition~\ref{RDCMdef} is
\begin{equation}\label{asc}
X_t\mid (X_T=0,X_{t_0}=x_{t_0})\sim \mathcal{N}(k_{t\mid t_0},\sigma^2_{t\mid t_0}),
\end{equation}
where
\begin{align}
k_{t\mid t_0}
&:= m_{t\mid t_0}
-\frac{e^{-a(T-t)}v_{t\mid t_0}}{v_{T\mid t_0}}\,m_{T\mid t_0},
\label{k}\\
\sigma^2_{t\mid t_0}
&:= v_{t\mid t_0}
-e^{-2a(T-t)}\frac{v_{t\mid t_0}^2}{v_{T\mid t_0}}.
\label{sigma0}
\end{align}
The bridge\footnote{Here ``bridge'' refers to the terminal conditioning \(X_T=0\). In the classical sense, a bridge is usually pinned at both endpoints.} variance $\sigma^2_{t\mid t_0}$ admits the following useful factorization.

\begin{mylemma}\label{Factor}
For any $t_0<t\le T$, $a>0$, and $\theta\in\Theta$,
\begin{equation}
\sigma^2_{t\mid t_0}=v_{t\mid t_0}\,R_{t\mid t_0}(a,\theta),
\label{sigma}
\end{equation}
where $R_{t\mid t_0}(a,\theta)
:=
\frac{v_{T\mid t}}
{e^{-2a(T-t)}v_{t\mid t_0}+v_{T\mid t}}
\in(0,1)$.
\end{mylemma}
\begin{myproof}
From \eqref{eq:var} we have \(v_{T\mid t_0}=e^{-2a(T-t)}v_{t\mid t_0}+v_{T\mid t}\) that plugged in \eqref{sigma0} immediately yields \eqref{sigma}.
\end{myproof}
We now introduce a mild set of requirements ensuring the existence of a maximizer of the exact bridge likelihood on a fixed observation grid.

\begin{myassumption}\label{ass:mle}
Let \(\Psi=\mathcal A\times\mathbb B\times\Theta\), with
\(\mathcal A\subset(0,\infty)\), \(\mathbb B\subset\mathbb R^{n_1}\), and
\(\Theta\subset\mathbb R^{n_2}\), be compact.
Assume that there exists $\delta>0$ such that the last observation time $t_n$ satisfies $t_n\le T-\delta$, and that:
\begin{enumerate}
    \item[(i)] the map \((t,b)\mapsto \phi(T-t,b)\) is jointly continuous on \([t_0,T-\delta]\times\mathbb B\);
    \item[(ii)] the map \((t,\theta)\mapsto g(T-t,\theta)\) is jointly continuous on \([t_0,T-\delta]\times\Theta\).
\end{enumerate}
\end{myassumption}
A natural choice for the parameter sets \(\mathbb{B}\) and \(\Theta\) in the RDCM specification \eqref{eq:phiandg} is
\begin{align*}
\mathbb{B} &:= [\underline{b}^{-},\overline{b}^{-}]\times[\underline{b}^+,\overline{b}^+],
\qquad \underline{b}^+>0,\\
\Theta &:= [\underline{\theta}^-,\overline{\theta}^-]\times[\underline{\theta}^+,\overline{\theta}^+],
\qquad \underline{\theta}^->0,\quad \underline{\theta}^+>0.
\end{align*}
This choice is assumed throughout the remainder of the paper whenever the specification in \eqref{eq:phiandg} is adopted.
\begin{myproposition}\label{prop:bridge_mle}
Let \(t_0<t_1<\cdots<t_n\le T-\delta\) be a fixed observation grid and $\psi \in \Psi$. Let
\(f_{t_i\mid t_{i-1}}(x_{t_i};\psi)\) denote the Gaussian transition density of the unconstrained version of the linear SDE in \eqref{RDCM} and let
\(f_{t_i\mid t_{i-1},T}(x_{t_i};\psi)\) denote the corresponding bridge  transition density, conditional on \(X_{t_{i-1}}=x_{t_{i-1}}\) and \(X_T=0\). Given \(X_{t_0}=x_{t_0}\) and \(X_T=0\), define
\[
\ell_n(\psi):=
\sum_{i=1}^n \log f_{t_i\mid t_{i-1},T}(x_{t_i};\psi).
\]
Under Assumptions \ref{rdcm_assumption} and \ref{ass:mle}, the following hold:
\begin{enumerate}
    \item[(i)] the bridge log-likelihood admits the telescopic representation
    \begin{equation}\label{eq:loglikSimpl}
    \ell_n(\psi)
    =
    \log f_{T\mid t_n}(0;\psi)
    -\log f_{T\mid t_0}(0;\psi)
    +\sum_{i=1}^n \log f_{t_i\mid t_{i-1}}(x_{t_i};\psi),
    \end{equation}
    where \(f_{T\mid t_i}(0;\psi)\) denotes the Gaussian density of \(X_T\), conditional on \(X_{t_i}=x_{t_i}\), evaluated at \(0\);
    
    \item[(ii)] the map \(\ell_n:\Psi\to\mathbb R\) is continuous. In particular,
    \[
    \arg\max_{\psi\in\Psi}\ell_n(\psi)\neq\varnothing.
    \]
\end{enumerate}
\end{myproposition}

\begin{myproof}
See Appendix \ref{prop:bridge_mle_Appendix}
\end{myproof}

\subsubsection{Weak identification with pre-\(\tau\) data}
\label{subsubSect}

A weak-identification issue arises when the observation window lies entirely before the decay date, namely when \(t_n < \tau\). In this case, the sample does not contain direct information on the part of the drift and diffusion coefficients that is active only on the post-\(\tau\) region. As a consequence, any parameter subvector affecting \(\phi(T-t,b)\) and \(g(T-t,\theta)\) only on \((\tau,T-\delta]\) can enter the likelihood only through the endpoint correction term in \eqref{eq:loglikSimpl}. For the parametric specification in \eqref{eq:phiandg}, this issue concerns the post-decay shape parameters $\left(b^{+},\theta^{+}\right)$.

For simplicity of the derivation, we assume the regularity and differentiability conditions needed to differentiate the bridge moments with respect to the post-\(\tau\) parameters and to exchange differentiation and integration whenever required.

When the sample does not overlap the decay region, the data do not observe the change in shape of the coefficients directly. The corresponding parameters are then informed only indirectly through the terminal bridge correction. This mechanism is too weak to identify the whole post-\(\tau\) block jointly as highlighted in the following lemma.

\begin{mylemma}\label{lem:ident_prob} When $\tau > t_n$, the parameters $b^{+}$ and $\theta^{+}$ in \eqref{eq:phiandg} are not jointly identifiable. Specifically, the first-order stationary conditions for the log-likelihood lead to a mathematical contradiction. 
\end{mylemma} 

\begin{myproof} 
See Appendix \ref{Appendix_subsubSect}
\end{myproof}

\begin{myremark}
Lemma \ref{lem:ident_prob} does not imply that the model is not useful on pre-\(\tau\) data. It only shows that one should not expect full identification of the post-\(\tau\) block from such a sample. In applications based only on pre-\(\tau\) observations, one should therefore impose a parsimonious specification for the post-\(\tau\) component by fixing, for instance $b^{+}=1$ in \eqref{eq:phiandg}. 
\end{myremark}

\subsection{Transition with daily regime switching}\label{SectioSRCDM}

While the RDCM framework provides a benchmark specification with a deterministic transition deadline, we now relax this assumption by allowing the perceived deadline to vary over time through a latent regime. The switching mechanism is introduced strictly on the observation lattice: the latent component is modeled as a discrete-time hidden Markov model (HMM) that selects the SDE followed by the observable process between two consecutive observation dates. We refer to this extended framework as the Switching Regulatory Deadline-Constrained Model (SRDCM).

Let $\mathcal{S} = \left\{1, \dots, m\right\}$ be the finite set of regimes and let $\bar{\Delta} >0$ denote a fixed observation step. The observation grid is defined by
\[
t_i = t_0 + i \bar{\Delta}, \qquad i = 0,1,\dots,n.
\]
Here $\bar{\Delta}$ is part of the model specification, since it determines the dates at which the latent regime is updated. Note that it should not be confused with the observation mesh introduced later for asymptotic estimation.

\noindent The latent regime is modeled as a discrete-time homogeneous Markov chain
\begin{equation}
    S^{\bar{\Delta}} := \left\{S^{\bar{\Delta}}_{i}\right\}_{i=0}^{n-1}, \qquad S^{\bar{\Delta}}_i \in \mathcal{S},
    \label{HMM_Switching}
\end{equation}
with initial distribution $\pi_0 = \left(\pi_{0,j}\right)_{j\in \mathcal{S}}$, where
\begin{equation}
\pi_{0,j} = \mathbb{P}\left(S^{\bar{\Delta}}_{0} = j\right),
\label{p0_SRDCM}
\end{equation}
and transition matrix $\mathbf{P} = \left(p_{hk}\right)_{h, k =1}^m$ with entries
\begin{equation}
p_{hk} = \mathbb{P}\left(S^{\bar{\Delta}}_{i+1}= k \mid S^{\bar{\Delta}}_{i}= h\right), \qquad i = 0, \dots, n-2.
\label{co}
\end{equation}
Each $p_{hk}$ represents the probability of moving from regime $h$ to regime $k$ between two consecutive observation dates. In particular, the off-diagonal entries $p_{hk}$, with $h\neq k$, describe shifts in the perceived transition deadline. Each regime $j\in\mathcal S$ is associated with its own parameter block $\left(a_j,b_j,\theta_j\right)$, structural time point $\tau_j$, and transition point $T_j$, hence with a specific RDCM-type dynamics.

Let $(\Omega, \mathcal{F}, \mathbb{P})$ be a probability space supporting a Wiener process $W := \left\{W_t\right\}_{t\geq 0}$ and the latent chain $S^{\bar{\Delta}}$. Let $\mathbb{F}^{W} = \{\mathcal{F}^{W}_{t}\}_{t \geq 0}$ be the natural filtration of $W$, with
\[
\mathcal{F}^{W}_{t} = \sigma (W_s : 0 \leq s \leq t),
\]
and let $\mathbb{F}^{S^{\bar{\Delta}}} := \{\mathcal{F}^{S^{\bar{\Delta}}}_{i}\}_{i = 0}^{n-1}$ be the discrete-time filtration generated by the latent chain, namely
\[
\mathcal{F}^{S^{\bar{\Delta}}}_{i}:=\sigma(S^{\bar{\Delta}}_{0},\dots,S^{\bar{\Delta}}_{i}),\qquad i=0,\dots,n-1.
\]
We equip $(\Omega, \mathcal{F}, \mathbb{P})$ with the enlarged filtration $\mathbb{G} := \left\{\mathcal{G}_{t}\right\}_{t \geq0 }$, where for any $t \in \left[t_i, t_{i+1}\right)$,
\[
\mathcal{G}_t = \mathcal{F}_t^W \vee \mathcal{F}_i^{S^{\bar{\Delta}}}.
\]
This reflects the timing convention of the model: the realization of $S^{\bar{\Delta}}_{i}$ at time $t_i$ determines the bridge dynamics over the subsequent interval $\left(t_i,t_{i+1}\right]$.

Conditional on $S^{\bar{\Delta}}_{i}= j$, $\left\{X_t\right\}_{\left(t_i,t_{i+1}\right]}$ evolves according to the regime-$j$ deadline-constrained bridge
\begin{equation}\label{SDE_SRDCM}
\begin{cases}
\mathrm{d} X_t^{(j)} = a_j\left[\phi_j(T_j-t,b_j)-X_t^{(j)}\right]dt +
g_j\left(T_j-t,\theta_j\right)\,\mathrm{d}W_t,
\qquad t\in(t_i,t_{i+1}],
\\
X_{t_i}^{(j)} = X_{t_i}, \quad \text{and} \quad X_{T_j}^{(j)} = 0.
\end{cases}
\end{equation}
Within each regime $j$, the process behaves as a diffusion conditioned to hit the terminal target $X_{T_{j}}=0$. Although this structure is mathematically equivalent to a stochastic bridge (see \citealt{doob1957conditional,doob1984classical}), only the terminal point is fixed by the regulatory deadline, while the starting point $x_{t_i}$ is inherited from the previous interval. 
This leads to the following definition.

\begin{mydefinition}\label{Def_SRDCM_MODEL}
A pair $\left(X, S^{\bar{\Delta}}\right)$ adapted to the enlarged filtration $\mathbb{G}$ is called a Switching Regulatory Deadline-Constrained Model if
\begin{enumerate}
\item $S^{\bar{\Delta}}$ is a discrete-time homogeneous Markov chain on $\mathcal{S}$;
\item for each interval $\left(t_i,t_{i+1}\right]$, conditional on $S^{\bar{\Delta}}_{i}=j$, the continuous component evolves according to \eqref{SDE_SRDCM};
\item the observed process $X$ is obtained by concatenating the regime-specific bridges along the realized latent regime sequence.
\end{enumerate}
\end{mydefinition}
Although the drift and diffusion coefficients of the SRDCM may exhibit jump discontinuities at the observation dates $t_i$ due to regime changes, the sample path of $X$ is almost surely continuous by construction.

For each fixed sample size $n$, the natural state space of the model is $\left(E_n,\mathcal{E}_n,\mathbb{P}_n\right)$ where
\[
E_n:=\mathbb{R}^{n+1}\times \mathcal{S}^n,
\qquad
\mathcal{E}_n:=\mathcal{B}(\mathbb{R}^{n+1})\otimes \mathcal{P}(\mathcal{S}^n),
\]
with $\mathcal{B}(\mathbb{R}^{n+1})$ the Borel $\sigma$-field on $\mathbb{R}^{n+1}$ and $\mathcal{P}(\mathcal{S}^n)$ the power set of $\mathcal{S}^n$. The probability measure $\mathbb{P}_n$ is absolutely continuous with respect to the reference measure $\mu_n$ defined as:
\begin{equation}
\mu _{n}=\nu ^{\otimes (n+1)}\times \lambda ^{\otimes n},
\label{referenceMeasures}
\end{equation}
where $\nu$ and $\lambda$ denote, respectively, the Lebesgue and counting measures. The joint law of $\left(X_{t_0}, \dots, X_{t_n}, S^{\bar{\Delta}}_0, \dots, S^{\bar{\Delta}}_{n-1} \right)
$ has a mixed continuous-discrete joint density representation due to the presence of continuous  $\left(X_{t_0},\ldots,X_{t_n}\right)$ and discrete $\left(S^{\bar{\Delta}}_0, \dots, S^{\bar{\Delta}}_{n-1} \right)$ coordinates. This provides the natural continuous-discrete framework for the likelihood construction where each one-step contribution combines a regime-transition probability with a regime-dependent bridge density.

Figure~\ref{Fig1} illustrates the interaction between the lattice-valued hidden sequence $\left\{S^{\bar{\Delta}}_{i}\right\}_{i=0}^{n-1}$ and the continuous process $X$. At each observation date $t_i$, the realization of $S^{\bar{\Delta}}_{i}$ selects a regime $j\in\mathcal S$, that is, a parameter block together with a target deadline $T_j$. The segment $\left\{X_t\right\}_{t \in (t_i, t_{i+1}]}$ is then generated according to the corresponding bridge dynamics starting from the inherited value $X_{t_i}$.

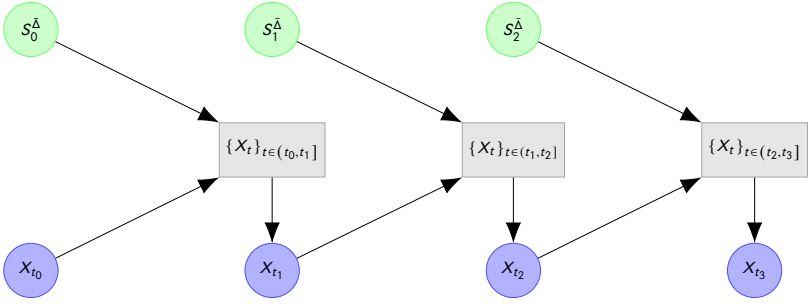
\begin{figure}[h!]
		\centering
	\resizebox{0.75\textwidth}{!}{ \begin{tikzpicture}[
		circleNode/.style={circle, draw=green!60, fill=green!20, minimum size=9mm},
		squareNode/.style={circle, draw=blue!70, fill=blue!30, minimum size=9mm},
		middleNode/.style={rectangle, draw=gray!70, fill=gray!20, minimum size=9mm},
		>={Latex[scale=2]}
		]
		
		\node[circleNode] (Z0) at (0,4) {$S^{\bar{\Delta}}_{0}$};
		\node[circleNode] (Z1) at (4,4) {$S^{\bar{\Delta}}_{1}$};
		\node[circleNode] (Z2) at (8,4) {$S^{\bar{\Delta}}_{2}$};
		
		\node[squareNode] (Y0) at (0,0) {$X_{t_0}$};
		\node[squareNode] (Y1) at (4,0) {$X_{t_1}$};
		\node[squareNode] (Y2) at (8,0) {$X_{t_2}$};
		\node[squareNode] (Y3) at (12,0) {$X_{t_3}$};
		
		\node[middleNode] (X1) at (4,2) {$\{X_t\}_{t \in \left(t_0, t_1\right]}$};
		\node[middleNode] (X2) at (8,2) {$\{X_t\}_{t \in \left(t_1, t_2\right]}$};
		\node[middleNode] (X3) at (12,2) {$\{X_t\}_{t \in \left(t_2, t_3\right]}$};
		
		\draw[->] (Z0) -- (X1.north west);
		\draw[->] (Y0) -- (X1.south west);
		\draw[->] (X1) -- (Y1);
		
		\draw[->] (Z1) -- (X2.north west);
		\draw[->] (Y1) -- (X2.south west);
		\draw[->] (X2) -- (Y2);
		
		\draw[->] (Z2) -- (X3.north west);
		\draw[->] (Y2) -- (X3.south west);
		\draw[->] (X3) -- (Y3);
		
	\end{tikzpicture}}
\caption{At each observation date $t_i$, the latent variable $S^{\bar{\Delta}}_{i}$ selects a regime $j \in \mathcal{S}$, namely a parameter block targeting deadline $T_j$. The segment $\{X_t\}_{t \in (t_i, t_{i+1}]}$ is then generated according to the corresponding bridge dynamics.}
\label{Fig1}
\end{figure}

In the empirical analysis, we restrict the observation horizon to
\[
[t_0,T_1-\delta], \qquad T_1:=\min_{j\in\mathcal S}T_j, \qquad \delta>0,
\]
so that all regimes remain active over the available sample. Equivalently, for fixed $\bar\Delta$, the sample size $n$ is bounded above by
\[
\bar n:=\max\left\{k\in\mathbb N:\ t_0+k\bar\Delta\leq T_1-\delta\right\}.
\]
This restriction is not part of the definition of the SRDCM itself, but an operational choice that is consistent with an economic interpretation: the date $T_1-\delta$ marks the onset of a pre-transition decision region in which the transition associated with the earliest deadline $T_1$ may already be perceived as nearly deterministic, because agents have undertaken costly and difficult-to-reverse actions that effectively presuppose its realization. Typical examples include large-scale adoption of high-efficiency buildings, electric mobility, electrified industrial equipment, and other technological adjustments that make a delayed transition increasingly implausible. From this point onward, the model is no longer used mechanically: either the transition associated with $T_1$ is effectively accepted, in which case the analysis stops, or the data suggest that such a transition has not materialized, in which case a new SRDCM is specified on an updated set of credible deadlines. In the remaining part of the paper, we refer to this approach as a sequential model updating strategy.

\subsubsection{Calibration discussion on SRDCM}\label{EMSRDCM}
This section discusses two main issues related to the fitting procedure of the SRDCM. The first is to show how the telescopic structure available in the RDCM is broken at each time \(t_i\) where the regime changes, thus reducing the weak-identification issue with pre-\(\tau\) data discussed in Section~\ref{subsubSect}. The second concerns the Forward-Backward recursions in the Baum-Welch algorithm (see \citealt{baum1970maximization} for details) used in the estimation procedure of the SRDCM. Indeed, due to the SRDCM structure described in Figure~\ref{Fig1}, these recursions are slightly modified with respect to the classical hidden Markov model case \citep[see][]{zucchini2009hidden}.

We assume that the process \(\{X_t\}_{t\ge 0}\) is observed only on the grid
\(t_0,t_1,\dots,t_n\), yielding the realizations \((x_{t_0},x_{t_1},\dots,x_{t_n})\).
For simplicity, we assume that the transition probability matrix \(\mathbf{P}\) is ergodic and has no absorbing
states, that is, the probabilities \((p_{hk})^{m}_{h,k=1}\) in \eqref{co} are all strictly
positive, and that the components \((\pi_{0,j})_{j\in\mathcal{S}}\) defined in
\eqref{p0_SRDCM} are strictly positive. Under these assumptions, all HMM paths
\(s^{\bar{\Delta}} \in \mathcal{S}^n\) are admissible and contribute to the computation of
the direct SRDCM log-likelihood conditioned to the first observation $x_{t_0}$.

Let
\(
\psi:=\Big(\left(a_j,b_j,\theta_j\right)_{j\in\mathcal{S}},(p_{hk})^{m}_{h,k=1},(\pi_{0,j})_{j\in\mathcal{S}}\Big)
\)
denote the vector of SRDCM parameters. For any
\(
s^{\bar{\Delta}}\in \mathcal{S}^{n},
\)
we introduce the log-likelihood 
\(
\ell_n^{}\left(\psi\mid s^{\bar{\Delta}} \right)
\)
conditional on the latent path \(s^{\bar{\Delta}}\). This is a natural generalization of
\eqref{eq:loglikSimpl}, since
\(
\ell_n^{}\left(\psi\mid s^{\bar{\Delta}} \right)
\)
can be viewed as a concatenation of RDCM log-likelihood contributions over the sub-intervals of the observation grid on which the regime remains constant. 
Let $N\left(s^{\bar{\Delta}}\right)$ denote the number of regime changes in a path $s^{\bar{\Delta}}$, we introduce the sequence of regime change times 
$\left\{\kappa_{r}\left(s^{\bar\Delta}\right)\right\}_{r=0}^{N(s^{\bar{\Delta}})+1}$ 
such that for $r=1,\ldots N(s^{\bar{\Delta}})$ the $r$-th change time is 
$\kappa_{r}\left(s^{\bar\Delta}\right)=t_i$ where $s_{t_{i}}\neq s_{t_{i-1}}$
with $\kappa_{0}\left(s^{\bar\Delta}\right)=t_0$ 
and $\kappa_{N\left(s^{\bar\Delta}\right)+1}\left(s^{\bar\Delta}\right)=t_n$. 
For each block \(r=1,\dots,N(s^{\bar\Delta})+1\), let
\[
j_r
:=
s^{\bar\Delta}_{\kappa_{r-1}(s^{\bar\Delta})}
\]
denote the regime active on the interval
\(
(\kappa_{r-1}(s^{\bar\Delta}),\kappa_r(s^{\bar\Delta})].
\)
Thus, \(
\ell_n^{}\left(\psi\mid s^{\bar{\Delta}} \right)
\) reads:
\footnotesize
\begin{equation}
\ell_n^{}\left(\psi\mid s^{\bar{\Delta}} \right)
=
\sum_{r=1}^{N\left(s^{\bar\Delta}\right)+1}
\sum_{t_i \in\left(\kappa_{r-1}(s^{\bar\Delta}), \kappa_{r}(s^{\bar\Delta})\right]}
\left[
\log\left(
\frac{
f_{T_{j_r}\mid t_i}
\left(0;\psi_{j_r}\right)
}{
f_{T_{j_r}\mid t_{i-1}}
\left(0;\psi_{j_r}\right)
}
\right)
+
\log
f_{t_i\mid t_{i-1}}
\left(
x_{t_i};
\psi_{j_r}
\right)
\right],
\label{eq:lcndatosdelta}
\end{equation}
\normalsize
where $f_{t_i\mid t_{i-1}}
\left(
x_{t_i};
\psi_{j_r}
\right)$ is the transition density of the linear SDE in \eqref{SDE_SRDCM} without the final condition and  $s_{t_{i-1}}^{\bar \Delta}=j_r$.
The formula in \eqref{eq:lcndatosdelta} highlights that $\ell^{c}_{}\left(\psi\mid s^{\bar \Delta}\right)$ is obtained by $N\left(s^{\bar{\Delta}}\right)+1$ log-likelihood RDCM blocks of the form in \eqref{eq:loglikSimpl}. Moreover, the telescopic part in each block $\left(\kappa_{r-1}\left(s^{\bar{\Delta}}\right),\kappa_{r}\left(s^{\bar{\Delta}}\right)\right]$ simplifies as:
\begin{equation}
    \sum_{t_i \in\left(\kappa_{r-1}(s^{\bar\Delta}), \kappa_{r}(s^{\bar\Delta})\right]}
\log\left(
\frac{
f_{T_{j_r}\mid t_i}
\left(0;\psi_{j_r}\right)
}{
f_{T_{j_r}\mid t_{i-1}}
\left(0;\psi_{j_r}\right)
}
\right)=\log\left(\frac{f_{T_{j_r}\mid \kappa_{r}\left(s^{\bar\Delta}\right)}
\left(0;\psi_{j_r}\right)}{f_{T_{j_r}\mid \kappa_{r-1}\left(s^{\bar\Delta}\right)}
\left(0;\psi_{j_r}\right)}\right).
\label{RegimewisetelescopicTerm}
\end{equation}
Using \textit{the log-sum-exp} approach, the SRDCM direct loglikelihood given the first observation $x_{t_0}$ has the final form:
\begin{equation}
    l_n\left(\psi\right):= \log\left[ \sum_{s^{\bar{\Delta}}\in \mathcal{S}^{n}}e^{\ell_{n}\left(\psi\mid s^{\bar\Delta}\right)}\mathbb{P}\left(s^{\bar\Delta}\right)\right]
    \label{eq:Direct_log_likSRDCM}
\end{equation}
where \(\mathbb P(s^{\bar\Delta})\) denotes the probability assigned to the latent path \(s^{\bar\Delta}\) by the hidden Markov chain on the grid. Therefore, the observed likelihood admits a \textit{log-sum-exp} representation.

\begin{myremark}
Considering pre-$\tau$ data, the issue highlighted in Section~\ref{subsubSect} is mitigated in the SRDCM, since the global telescopic term of the RDCM is fragmented into regime-wise telescopic blocks of the form \eqref{RegimewisetelescopicTerm}.
\end{myremark}

Although the direct observed likelihood \(\ell_n(\psi)\) in
\eqref{eq:Direct_log_likSRDCM} is well defined, its numerical maximization is
cumbersome because it involves a sum over all admissible hidden paths. Thus calibration is carried out by means of the Baum-Welch algorithm% \citep[see][]{baum1970maximization}
, that is, the Expectation--Maximization
procedure (see \citet{dempster1977maximum} for details) naturally associated with a hidden Markov
model. In the SRDCM setting, however, the forward and backward recursions are
not the standard ones, because \(S_{i-1}^{\bar\Delta}\) governs the evolution of
\(X_t\) over the subsequent interval \((t_{i-1},t_i]\), as illustrated in
Figure~\ref{Fig1}. For the general Baum-Welch framework we refer to
\cite{zucchini2009hidden}. 

We define the forward $\alpha_i(j)$ and backward $\beta_i(j)$ quantities with $i=1,\ldots,n$ and $j\in \mathcal{S}$ as the usual joint density-probability objects associated with the mixed continuous-discrete structure introduced Section \ref{SectioSRCDM}.\footnote{A  theoretical measure construction can be obtained from the reference measure $\mu_n$
in \eqref{referenceMeasures}; we omit it here and refer to \cite[][Chapter 4]{cinlar2011probability}.} Specifically, $\alpha_i(j)$ is the joint density-probability of $\left(x_{t_1},\ldots,x_{t_{i-1}}, S^{\bar \Delta}_{i-1}=j,x_{t_i}\right)$
given $x_{t_0}$ while $\beta_i\left(j\right)$ is the conditional density-probability of $\left(x_{t_{i+1}},\ldots,x_{t_n}\right)$ given $\left(x_{t_i},S^{\bar \Delta}_{i-1}=j\right)$. These two quantities can be used to compute the posterior probabilities $\left\{\gamma_{i-1}(j)\right\}_{j\in\mathcal{S}}$ with $i=1,\ldots,n$ used in the E-step. Specifically, they read
\begin{equation}
\gamma_{i-1}(j)
=
\mathbb{P}(S_{t_{i-1}}=j\mid X_{t_0}=x_{t_0},\dots,X_{t_n}=x_{t_n})
=
\frac{\alpha_i(j)\beta_i(j)}
{\sum_{k=1}^m\alpha_i(k)\beta_i(k)}.
\label{gammaPosterior}
\end{equation}
We refer to \cite{zucchini2009hidden}, for integrating \eqref{gammaPosterior} in the E-step of the classical Expectation-Maximization algorithm developed by \cite{dempster1977maximum}.   
We complete this section with the following proposition that shows the structure of the recursions for the forward $\alpha_i(j)$ and backward $\beta_i(j)$ quantities.
\begin{myproposition}
\label{prop:FB_SRDMC}
Let
\[
f_{t_{i}\mid t_{i-1},T_j}(x_{t_i};\psi) :=  \frac{f_{t_i\mid t_{i-1}}\left(x_{t_i};\psi_j\right)f_{T_j\mid t_i}\left(0;\psi_j\right)}{f_{T_j\mid t_{i-1}}\left(0;\psi_j\right)}
\]
denote the regime-$j$ bridge transition density over $(t_{i-1},t_i]$ with $f_{t_i\mid t_{i-1}}\left(x_{t_i};\psi_j\right)$ the regime-wise version transition density  of a linear SDE without final condition, and let $p_{hj}$
be the transition probability in \eqref{co} from regime $h$ to regime $j$. 
Introducing the compact notation $f_j(x_{t_i}\mid x_{t_{i-1}}):=f_{t_{i}\mid t_{i-1},T_j}(x_{t_i};\psi_j)$, we have
\[
\alpha_1(j)=\pi_{0,j}\,f_j(x_{t_1}\mid x_{t_0}),
\qquad
\beta_n(j)=1,
\]
and, for $i=2,\dots,n$,
\begin{equation}
\alpha_i(j)
=
f_j(x_{t_i}\mid x_{t_{i-1}})
\sum_{h=1}^m\alpha_{i-1}(h)p_{hj},
\label{eq:alpha_recursion_our}
\end{equation}
while, for $i=n-1,\dots,1$,
\begin{equation}
\beta_i(j)
=
\sum_{k=1}^m p_{jk}\,f_k(x_{t_{i+1}}\mid x_{t_i})\,\beta_{i+1}(k).
\label{eq:beta_final}
\end{equation}
\end{myproposition}

\begin{myproof}
See Appendix~\ref{app:proof_prop32}.
\end{myproof}

\section{Empirical Analysis}\label{empanalysis}
In this section, we conduct an empirical analysis on daily data observed from 8 September 2021 to 17 January 2025. The analysis is based on  the following three steps: calibration of model parameters, residuals analysis and model selection. 

We first start with the simplest formulation, i.e., we calibrate the RDCM model  conditioned to the assumption that transition will occur at  1 January 2031 with $\phi(\cdot,\cdot)$ and $g(\cdot ,\cdot)$ defined in \eqref{eq:phiandg}. $\tau$ is set to 8 September 2029, i.e., the date on which  both $\phi$ and $g$ start to decrease. Fitted results are presented in Table \ref{TabEstRCDM} in which standard deviations (sd) are computed using the parametric bootstrap\footnote{In practice, we fit the model parameters to the observed series for $X_t$ and use them to simulate the process 1000 times. For each simulated path, we estimate the model parameters again and finally compute their standard deviation.} methodology with 1000 replications \citep[see][for details on the parametric bootstrap]{efron1994introduction}.
% \begin{table}[!htbp]
% \centering
% \begin{tabular}{lcc}
% \toprule
% Param.  &  $est.$  & sd\\   
% \midrule
% $a$   & 18.3233 & 3.4245 \\
% $b^-$ & 0.0489  & 0.0062\\
% %$b_2$ & 0.9916 & 0.0103\\
% $\theta^-$ & 0.2180 & 0.0053\\
% $\theta^+$ & 1.0033 & 0.0099\\
% \midrule
% AIC & -5128.259 &\\
% \bottomrule
% \end{tabular}
% \caption{Calibrated parameters for the RCDM and corresponding standard deviations obtained through bootstrapping.\label{TabEstRCDM}}
% \end{table}

%To measure the ability of the RCDM model to describe the behaviour of the greenium difference $X_t$ we construct the series of the residuals $\hat{Z}_{t_i}$ defined as:
In order to assess the accuracy of the RDCM model, we first compute the filtered residuals as:
\begin{equation}
    \hat{z}_{t_i}=\frac{x_{t_{i}}-\hat{k}_i}{\hat{\sigma}_i}
\end{equation}from fitted model parameters
with $\hat{k}_i$ and $\hat{\sigma}_i$ defined respectively as in \eqref{k} and \eqref{sigma}. From Figure \ref{ResidualsRCDM}, confirmed also by the $p$-value of the KS-test, we observe that residuals $\left\{\hat{z}_i\right\}_{i=1}^{n}$ appear far from  normally distributed; thus,  the RDCM model does not seem adequate to  describe the dynamics of $X_t$ over the considered period. 
% \begin{figure}[!htbp]
%         \centering
%         \includegraphics[width=0.65\textwidth]{img_paper/Residuals_RCDM_VasicekGaussianModel.eps}  % Sostituisci con il percorso della tua immagine
%         \caption{Residuals of a RCDM model. KS test $\text{D} = 0.1142$, $\text{$p$-value} = 3.444\times10^{-10}$ i.e. the null hypothesis of the KS test for the gaussianity of the residuals is rejected at the 95\% confidence level.  \label{ResidualsRCDM}}
% \end{figure}

\begin{figure}[!htbp]
\centering

\begin{minipage}{0.48\textwidth}
\centering
\begin{tabular}{lcc}
\toprule
Param.  &  est.  & sd\\   
\midrule
$a$   & 18.3233 & 3.4245 \\
$b^-$ & 0.0489  & 0.0062\\
$\theta^-$ & 0.2180 & 0.0053\\
$\theta^+$ & 1.0033 & 0.0099\\
\midrule
AIC & -5128.259 &\\
\bottomrule
\end{tabular}
\vspace{1cm}
\captionof{table}{Calibrated parameters (est.) for the RDCM and corresponding standard deviations (sd) obtained through bootstrapping.}
\label{TabEstRCDM}
\end{minipage}
\hfill
\begin{minipage}{0.48\textwidth}
\centering
\includegraphics[width=\textwidth]{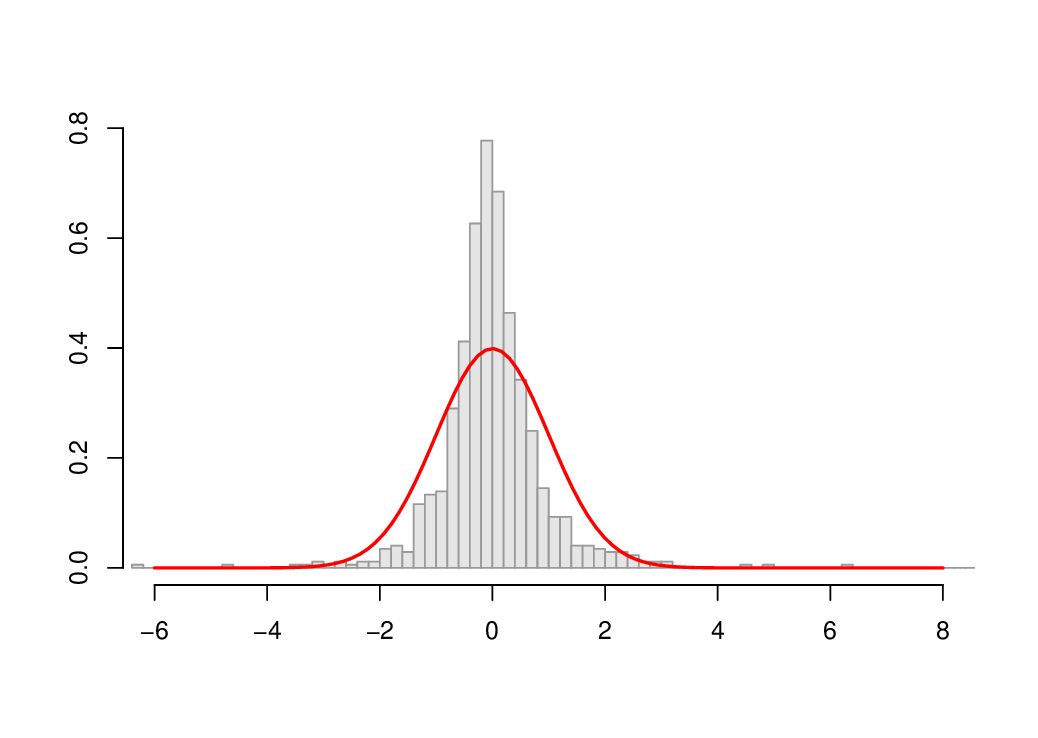}
\vspace{-1.15cm}
\captionof{figure}{Residuals of a RDCM model. KS test $\text{D} = 0.1142$, $p$-value $= 3.444\times10^{-10}$ i.e. the null hypothesis of the KS test for gaussianity of the residuals is rejected at the 95\% confidence level.}
\label{ResidualsRCDM}
\end{minipage}

\end{figure}

We then fit the switching RDCM model introduced in Section \ref{SectioSRCDM} 
to the same observed data. The maximization of $\ell_n\left(\phi\right)$ in \eqref{eq:Direct_log_likSRDCM} is performed by Expectation-Maximization procedure  based on the Baum-Welch algorithm \citep[see][for instance]{baum1970maximization,casgrain}. Table \ref{TabEst2Comp} reports the fitted parameters for a two-state regime switching RDCM model. In the first scenario the model is the RDCM model considered in the previous analysis, i.e., transition occurs within January 1, 2031 with structural point $\tau_1$ set to 8 September 2029 while in the second scenario the transition will happen at 1 January 2041 and the corresponding $\tau_2$ is set to 04 September 2036. %In Table \ref{TabEst2Comp}, the bootstrapped standard errors for all parameters are also reported. 

\begin{table}[!htbp]
\centering
\resizebox{1\textwidth}{!}{ %
\begin{tabular}{lcccccccccccccc}
\toprule
 & $a_1$ & $a_2$ & $b^{-}_{1}$ & $b^{-}_{2}$ & $b^{+}_{1}$ & $b^{+}_{2}$ & $\theta^{-}_{1}$ & $\theta^{-}_{2}$ & $\theta^{+}_{1}$ & $\theta^{+}_{2}$ & $p_{1,1}$ & $p_{2,1}$ & $p_{1,2}$ & $p_{2,2}$ \\
\midrule
est. & 33.9943 & 7.0034 & 0.0755 & 0.0363 & 1.0033 & 0.9918 & 0.3836 & 0.1002 & 1.0058 & 0.9969 & 0.8928 & 0.0457 & 0.1072 & 0.9543 \\
sd   & 2.8601  & 1.9549 & 0.0115 & 0.0110 & 0.0106 & 0.0099 & 0.0181 & 0.0032 & 0.0098 & 0.0098 & 0.0170 & 0.0078 & 0.0170 & 0.0078 \\
\midrule
AIC & \multicolumn{14}{c}{-5616.864} \\
\bottomrule
\end{tabular}}
\caption{Calibrated parameters and corresponding standard deviations obtained with 1000 bootstrapped samples.\label{TabEst2Comp}}
\end{table}
% \begin{table}[!htbp]
% \centering
% \resizebox{0.3\textwidth}{!}{ %
% \begin{tabular}{lcc}
% \toprule
% Param.  &  $est.$  & sd\\   
% \midrule
% $a_1$  & 33.9943 & 2.8601\\
% $a_2$  & 7.0034  & 1.9549   \\
% $b^{-}_{1}$ & 0.0755 & 0.0115 \\
% $b^{-}_{2}$ & 0.0363 & 0.0110 \\
% $b^{+}_{1}$ & 1.0033 & 0.0106 \\
% $b^{+}_{2}$ & 0.9918 & 0.0099\\
% $\theta^-_{1}$ & 0.3836 & 0.0181 \\
% $\theta^{-}_{2}$ & 0.1002 & 0.0032 \\
% $\theta^{+}_{1}$ & 1.0058 & 0.0098\\
% $\theta^{+}_{2}$ & 0.9969 & 0.0098\\
% $p_{1,1}$ & 0.8928 & 0.0170\\
% $p_{2,1}$ & 0.0457 & 0.0078\\
% $p_{1,2}$ & 0.1072 & 0.0170\\
% $p_{2,2}$ & 0.9543 & 0.0078\\
% \midrule
% AIC & -5616.864 &\\
% \bottomrule
% \end{tabular}}
% \caption{Calibrated parameters and corresponding standard deviations obtained with 1000 bootstrapped samples.\label{TabEst2Comp}}
% \end{table}

\noindent Comparing the AIC values in Tables \ref{TabEstRCDM} and \ref{TabEst2Comp} the switching two-regime model performs better than the RDCM. To further control the ability of the switching regime RDCM model to capture the behaviour of the observed data $\left\{x_{t_i}\right\}_{i=1}^n$, we conduct an analysis on residuals obtained with the following procedure:
\begin{itemize}
\item We assign the observation $x_{t_i}$ using the local decoding algorithm described in \cite{zucchini2009hidden}. Indeed, given the calibrated $\hat{\psi}_n$ parameters obtained by the sequence $\left\{x_{t_i}\right\}_{i=0}^{n}$, the fitted path $\hat{s}^{\bar{\Delta}}$ is obtained as a collection of indexes $\left\{\hat{j}_i\right\}_{i=1}^{n}$ whose $i$-th element $\hat{j}_i$ is the solution of the following problem:
\begin{equation}
\hat{j}_i:=\underset{j \in \left\{1,2\right\}}{\text{argmax}} \ \gamma_{t_{i-1}}\left(j; \hat{\psi}_n\right),
\label{algoritmForAssigned}
\end{equation}
 for each $i =1,\ldots, n$ and the map $j\mapsto\gamma_{t_{i-1}}\left(j;\cdot\right)$ has the form in \eqref{gammaPosterior}. 
\item Given  $\hat{s}^{\bar{\Delta}}$, the fitted residual $\hat{z}_{t_i}$ is obtained as:
\[
\hat{z}_{t_i}=\frac{x_{t_i}-k_i\left(\hat{j}_i;\hat{\psi}_n\right)}{\sigma_i\left(\hat{j};\hat{\psi}_n\right)}, 
\]
where $k_i\left(\cdot,\cdot\right)$ and $\sigma^2_i\left(\cdot,\cdot\right)$ are  in \eqref{k} and \eqref{sigma} respectively.  
\end{itemize}

As done in the RDCM model, a graphical comparison between its empirical density and the standard normal distribution is performed. Figure~\ref{Residuals} shows the residuals are closer to normality with respect to those obtained using the RDCM (see Figure \ref{ResidualsRCDM}). 
\begin{figure}[!h]
        \centering
        \includegraphics[width=0.65\textwidth]{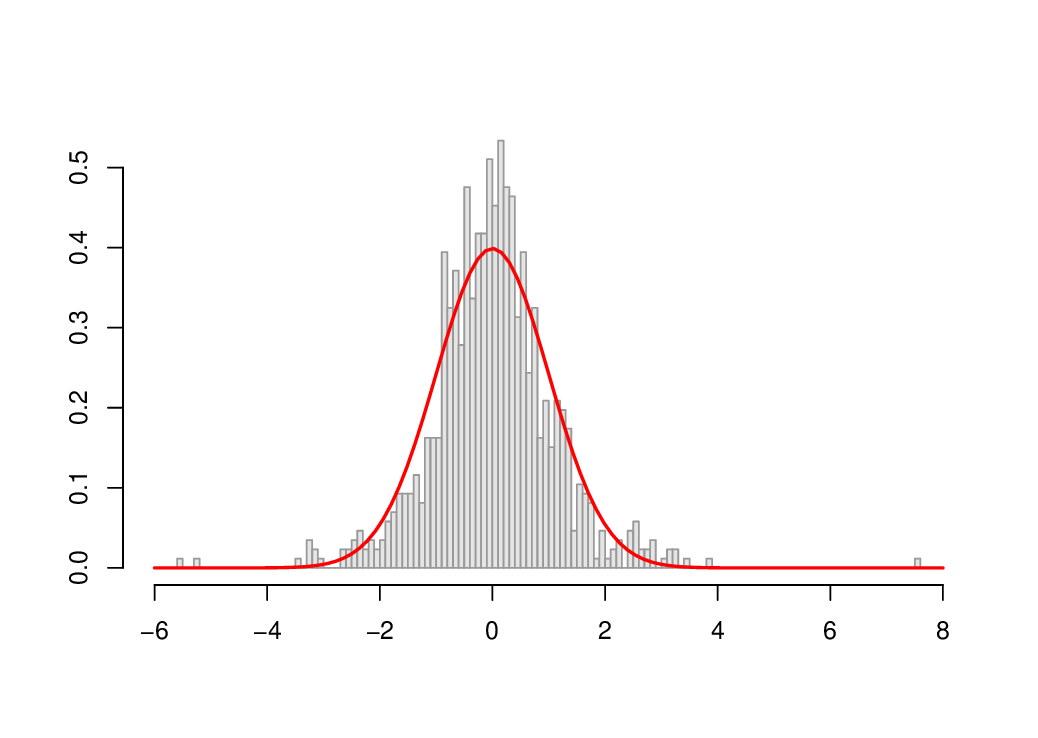}  % Sostituisci con il percorso della tua immagine
        \caption{Residuals of a two-state regime-switching RDCM model. KS test,   $\text{D} = 0.0422$, $\text{$p$-value} = 0.0933$
        %$\text{D} = 0.04004$, $\text{p-value} = 0.1261$  
        \label{Residuals}}
\end{figure}

The Kolmogorov-Smirnov test on the filtered residuals \(\{ \hat{z}_{t_i} \}\) returns a \(p\text{-value} = 0.093\), so we do not reject normality at the 5\% significance level. This suggests that the overall cumulative distribution function (CDF) of the residuals is reasonably close to the normal distribution.

\begin{figure}[!h]
    \centering
    \includegraphics[width=1\textwidth]{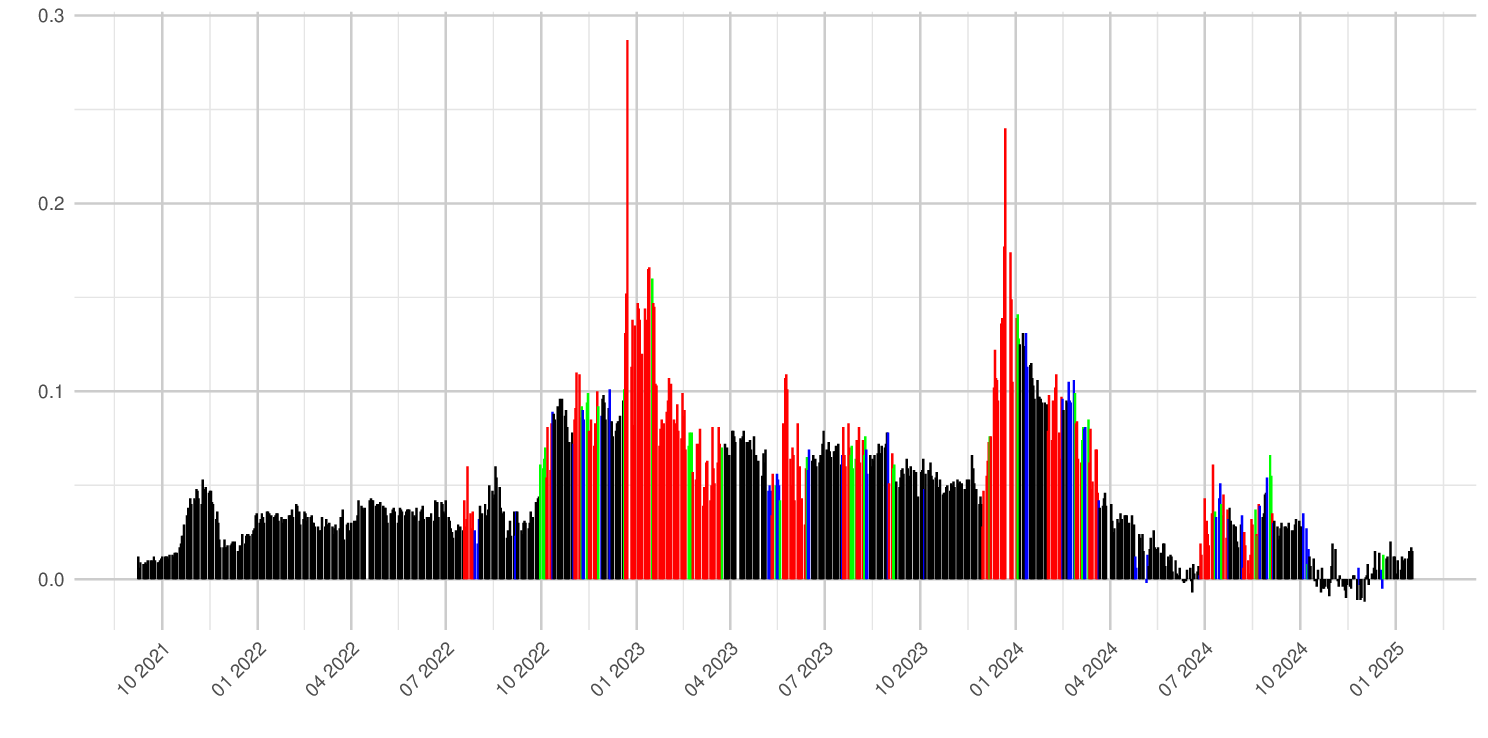}  % Sostituisci con il percorso della tua immagine
    \caption{Scenario detection with the two-state regime switching RDCM model applied to the observed data of $X_t$ ($y$-axis). Vertical red line corresponds to the case in which the transition within the end of year 2030 (first scenario) has $\gamma_{i-1}\left(1,\hat{\psi}_n\right)$ higher than 0.75. Vertical black line corresponds to case where the $\gamma_{i-1}\left(2,\hat{\psi}_n\right)$ is higher than 0.75 where scenario 2 denotes transition date of January 1, 2041. Green vertical line means that $\gamma_{i-1}\left(1,\hat{\psi}_n\right) \in [0.5,0.75)$ while the blue vertical lines are used for observation $x_{t_i}$ with $\gamma_{i-1}\left(2,\hat{\psi}_n\right) \in [0.5,0.75)$.  \label{ScenarioDetection}}
\end{figure}
Figure \ref{ScenarioDetection} shows the results of the scenario detection. For the German market, scenario 1, corresponding to transition by the end of 2030, does not appear to be perceived as realistic for most observations, especially during the last four months.

%Specifically, the red  and green (with $ \gamma_{i-1}\left(1;\hat{\psi}_n\right) \in [0.5,0.75)$) indicate that the observed data (most probable for $\hat{j}_{i-1}=1$), $x_{t_i}$ are classified, in the filtering exercise, as associated with the scenario whose deadline is January 1st, 2031, while black and blue indicate that the data will be assigned to the second scenario.

\section{$n$-Infill asymptotics and diffusion-block identification}
\label{sec:highfreq}
% =========================================================
% Section 5 + Appendix for the high-frequency identification part
% Requires: amsmath, amssymb, amsthm, mathtools, bm, hyperref
% =========================================================
%\section{$n$-Infill asymptotics and diffusion-block identification}
%\label{sec:highfreq}
In this section, we identify observation schemes under which the regime-wise
diffusion parameters can be structurally identified. Specifically, under mild
conditions, we prove that these parameters can be consistently estimated on their
visible part.\footnote{In the RDCM, only \(\theta^{-}\) is visible before \(\tau\),
whereas \((\theta^{-},\theta^{+})\) becomes visible once the observation window
overlaps the post-\(\tau\) region; the same interpretation applies regime by regime
in the SRDCM.}
The result specifies when the diffusion parameters carrying information about
competing transition dates are statistically identifiable. Proofs are collected in
Appendix~\ref{app:proofs_highfreq}.

% In this section, we ask under which observation scheme the identification of the perceived transition-time regime can be given a structural interpretation for the RDCM/SRDCM framework. The previous fitting procedure remains useful as an empirical filtering device, especially on expanding windows. Here, however, we consider a fixed-horizon \(n\)-infill scheme, under which the diffusion block becomes structurally visible once sufficiently high-frequency observations are available after the first structural time point \(\tau\) for RDCM ($\tau_1$ for SRDCM). Our goal is to show that, under mild conditions, the diffusion blocks can be consistently identified on their visible part \footnote{Here, ``visible part'' refers to the portion of the diffusion block that is effectively activated at the observed times. In the RDCM, only \(\theta^{-}\) is visible for \(t<\tau\), whereas both \((\theta^{-},\theta^{+})\) are visible for \(t\ge \tau\); see Assumption~\ref{ass:tau_activation}. The same interpretation applies regime by regime in the SRDCM.}
% , together with the temporal labels attached to them. 
% Proofs are collected in Appendix~\ref{app:proofs_highfreq}.

\subsection{RDCM: identification of the visible diffusion block}
\label{subsec:rdcm_highfreq}
%\input{inputs/Lorenzo/n_infill_old_version}

% The aim of this section is twofold. First,
% we construct a sequence of contrast functions $\{M_n\}_{n\geq 1}$ having, for every $n\geq 1$, the same maximizers as the log-likelihood $\ell_n(\psi)$ in \eqref{eq:loglikSimpl}. Second, we study their convergence on the compact observation interval $[t_0,l]$, with $l<T-\delta$, in order to apply the consistency theorem for $M$-estimators in \cite{van1998asymptotic}, Chapter~5.
The aim of this section is twofold. We first replace the exact bridge log-likelihood \(\ell_n(\psi)\) in \eqref{eq:loglikSimpl} with a sequence of contrast functions \(\{M_n\}_{n\geq 1}\) having, for every \(n\geq 1\), the same maximizers. This reduction is motivated by the fixed-horizon \(n\)-infill scheme on the compact interval \([t_0,l]\), with \(l<T-\delta\), under which \(M_n\) provides the natural normalization of \(\ell_n\) for asymptotic analysis. We then study the convergence of \(\{M_n\}_{n\geq 1}\) in order to apply the consistency theorem for \(M\)-estimators in \citet[ Chapter~5]{van1998asymptotic}.

Before stating the first reduction step, we recall the economic meaning of the region $(T-\delta,T]$. Its mathematical role will become clear in the proof of Proposition~\ref{prop:rdcm_uniform}.

\begin{myremark}
The interval $(T-\delta,T]$ is the pre-transition decision region, in which the deadline $T$ is already perceived as nearly deterministic by most economic agents, who have adjusted their behavior accordingly in view of the expected transition at time $T$. See the end of Section~\ref{SectioSRCDM} for an extension of this interpretation in the SRDCM framework.
\end{myremark}
The relation between \(M_n(\psi)\) and \(\ell_n(\psi)\) is stated in Corollary~\ref{rel_M_l_n}. For each \(n\), the observation grid is
\begin{equation}
t_i=t_0+i\Delta_n,\qquad i=0,\dots,n,\qquad \Delta_n:=\frac{l-t_0}{n},
\label{eq:grid_hf_rdcm}
\end{equation}
with \(l<T-\delta\). In the corollary below, \(M_n\) is written in terms of an observed sequence of realizations \(\{x_{t_i}\}_{i=1}^n\), using the same notation as in Section~\ref{Estim_RDCM_SUB}.  In Proposition~\ref{prop:rdcm_uniform}, by contrast, \(\{M_n\}_{n\ge1}\) will be viewed as a sequence of random contrast functions, since each \(M_n\) depends on the sample \(\{X_{t_i}\}_{i=1}^n\) generated by \eqref{RDCM}.

% \begin{mycorollary}[Relation between \(M_n(\psi)\) and \(\ell_n(\psi)\)]\label{rel_M_l_n}
\begin{mycorollary}\label{rel_M_l_n}
Define
\begin{equation}
M_n(\psi)
:=
\frac{1}{n}\log\!\Big(\frac{f_{T\mid t_n}(0;\psi)}{f_{T\mid t_0}(0;\psi)}\Big)
-\frac{1}{2n}\sum_{i=1}^n \log\!\Big(\frac{v_{t_{i}\mid t_{i-1}}}{\Delta_n}\Big)
-\frac{1}{2n}\sum_{i=1}^n \frac{(x_{t_i}-m_{t_i\mid t_{i-1}})^2}{v_{t_i\mid t_{i-1}}},
\label{observed_M_n_0}
\end{equation}
with $m_{t_i\mid t_{i-1}}$ and $v_{t_i\mid t_{i-1}}$ in \eqref{eq:mean} and \eqref{eq:var} respectively. Under Assumptions~\ref{rdcm_assumption} and \ref{ass:mle}, for any \(n\ge1\) and the observation scheme \eqref{eq:grid_hf_rdcm},
\begin{equation}
\arg \max_{\psi \in \Psi} \ell_n(\psi)
=
\arg \max_{\psi \in \Psi} M_n(\psi).
\label{equivalent_maximizers}
\end{equation}
\end{mycorollary}

\begin{myproof}
By \eqref{eq:loglikSimpl}, we have  $M_n(\psi) =
\frac{1}{n}
\ell_n(\psi)
+\frac{1}{2}\log(2\pi\Delta_n)
$, with standard manipulations.
 Hence, \(M_n\) and \(\ell_n\) have the same maximizers. 
\end{myproof}
While \(\ell_n\) and \(M_n\) have the same maximizers, under the \(n\)-infill scheme \(\Delta_n\to0\) the term \(\frac{1}{2}\log(2\pi\Delta_n)\) diverges to \(-\infty\); for this reason, \(M_n\) is the natural object for the asymptotic analysis below.
 
From this point onward, the contrast functions are viewed as a sequence of random functions, since for each \(n\) they depend on the sample \(\{X_{t_i}\}_{i=1}^n\), where \(X_{t_0}\) is given and the subsequent observations are recursively generated by the RDCM dynamics on the interval \([t_0,l]\). 

As anticipated above, the asymptotic identification on \([t_0,l]\) is driven by the diffusion block. Accordingly, throughout this section the drift parameters \((a,b)\) are treated as fixed, and the contrast is written as a function of \(\theta\in\Theta\) only. This convention is adopted for simplicity.

Once \((a,b)\) are fixed, we use the following notational convention for the one-step moments:
\[
m_{t_i\mid t_{i-1}}:=m_{t_i\mid t_{i-1}}(X_{t_{i-1}}),
\qquad
v_{t_i\mid t_{i-1}}:=v_{t_i\mid t_{i-1}}(\theta).
\]
Depending on the context, we shall use either the compact notation \(m_{t_i\mid t_{i-1}},\,v_{t_i\mid t_{i-1}}\) or the explicit notation \newline \(m_{t_i\mid t_{i-1}}(X_{t_{i-1}}),\,v_{t_i\mid t_{i-1}}(\theta)\). Under this convention, the contrast function \eqref{observed_M_n_0} takes the form
\begin{equation}
M_n(\theta)
:=
\frac{1}{n} R_n(\theta)
-\frac{1}{2n}\sum_{i=1}^n \log\!\left(\frac{v_{t_i\mid t_{i-1}}(\theta)}{\Delta_n}\right)
-\frac{1}{2n}\sum_{i=1}^n \frac{\Big(X_{t_i}-m_{t_i\mid t_{i-1}}(X_{t_{i-1}})\Big)^2}{v_{t_i\mid t_{i-1}}(\theta)},
\label{eq:Mn_rdcm_main}
\end{equation}
where
\begin{equation}
R_n(\theta)
:=
\log\!\Big(\frac{f_{T\mid t_n}(0;\theta)}{f_{T\mid t_0}(0;\theta)}\Big).
\label{eq:Rn_rdcm_main}
\end{equation}

We now study the consistency of the estimators $\hat{\theta}_n \in \arg \underset{\theta\in \Theta}{\max} M_{n}\left(\theta\right)$ as $n\rightarrow +\infty$. To this end, we impose the following assumptions.

% \begin{myassumption}[RDCM high-frequency regularity]
\begin{myassumption}
\label{ass:rdcm_hf}
The following conditions hold:
\begin{enumerate}
    \item $\Theta$ is compact and $\theta_0\in\mathrm{int}(\Theta)$ where $\theta_0$ denotes the true diffusion parameter vector;
    \item (\emph{gap condition}) the observation horizon satisfies $l\leq T-\delta$  
    for some $\delta>0$;
    \item (\emph{uniform ellipticity}) there exists $c_g>0$ such that
    \[
    g^2(T-t,\theta)\geq c_g,
    \qquad \forall (t,\theta)\in [t_0,T-\delta]\times\Theta;
    \]
    \item the map $(t,\theta)\mapsto g(T-t,\theta)$ is:
    
    \begin{enumerate}
    \item[a.] continuous on
    $[t_0,T]\times\Theta$; \item[b.] uniformly Lipschitz on $[t_0,T-\delta]\times\Theta$ in both arguments.
\color{black}
\end{enumerate}
    \item the drift block is fixed and satisfies the regularity conditions ensuring the
    existence of the corresponding linear bridge moments.
\end{enumerate}
\end{myassumption}
Differentiability of \(g(T-t,\theta)\) at \(t=\tau\) is not required, since the change of slope at the structural time point may destroy it; uniform Lipschitz continuity is sufficient below.

% \begin{myassumption}[Identifiability of the diffusion block in the RDCM]
\begin{myassumption}
\label{ass:rdcm_ident}
Define
\begin{equation}
M_\infty(\theta)
:=
-\frac{1}{2(l-t_0)}
\int_{t_0}^{l}
\left\{
\log\!\big(g^2(T-u,\theta)\big)
+
\frac{g^2(T-u,\theta_0)}{g^2(T-u,\theta)}
\right\}\,du.
\label{eq:Minfty_rdcm_main}
\end{equation}
Assume that $M_\infty$ admits a unique maximizer at $\theta_0$ $\forall l\in\left(t_0,T-\delta\right]$.\footnote{For the sake of brevity, Assumption~\ref{ass:rdcm_ident} is stated directly in terms of the limit contrast function $M_\infty(\theta)$. However, since the integrand in \eqref{eq:Minfty_rdcm_main} involves the U-shaped function $y \mapsto \log(y) + c/y$ for each fixed $c>0$, this requirement could be derived from a primitive assumption on $g(T-u, \theta)$; for instance, it is satisfied if the map $\theta \mapsto g(T-u, \theta)$ is injective for almost all $u \in (t_0, l]$.}
\end{myassumption}

The uniqueness condition is understood on the visible part of the diffusion block over \([t_0,l]\): on \(\theta^{-}\) if \(l\le\tau\), and on $(\theta^{-},\theta^{+})$ if \(l>\tau\).

% \begin{myproposition}[Consistency on the visible interval]\label{prop:rdcm_uniform}
\begin{myproposition}\label{prop:rdcm_uniform}
Suppose that Assumptions \ref{ass:rdcm_hf} and \ref{ass:rdcm_ident} hold, and that
\(l>\tau\), where $l$ is the last time in grid \eqref{eq:grid_hf_rdcm}. Then there exists a deterministic continuous function
\(M_\infty:\Theta\to\mathbb{R}\) such that
\begin{equation}
\sup_{\theta\in\Theta}\big|M_n(\theta)-M_\infty(\theta)\big|
\overset{\mathbb P}{\longrightarrow}0.
\label{eq:ConvergUnPofMn_for RDCM}
\end{equation}
If, in addition, \(M_\infty\) admits a unique maximizer \(\theta_0\), then any measurable maximizer
\[
\hat\theta_n\in\arg\max_{\theta\in\Theta}M_n(\theta)
\]
satisfies
\[
\hat\theta_n\overset{\mathbb P}{\longrightarrow}\theta_0.
\]
\end{myproposition}
\begin{myproof}
See Appendix~\ref{proof:rdcm_uniform}.
\end{myproof}

% \begin{myremark}[The case \(l\le \tau\)]\label{rem:rdcm_pre_tau}
\begin{myremark}\label{rem:rdcm_pre_tau}
If \(l\le\tau\), then by Assumption \ref{ass:tau_activation} the data-dependent part of
\(M_n(\theta)\) depends only on \(\theta^{-}\). More precisely, on \([t_0,l]\) one has
\[
g(T-t,\theta)=g(T-t,\theta^-),
\qquad t\in[t_0,l],
\]
so that the contrast reduces to a function of \(\theta^{-}\) only. We denote the corresponding reduced contrasts by
\(M_n(\theta^-)\) and \(M_\infty(\theta^-)\).
\end{myremark}

% \begin{mycorollary}[Consistency before the structural time point]\label{cor:rdcm_pre_tau}
\begin{mycorollary}\label{cor:rdcm_pre_tau}
For \(l\le\tau\), under the assumptions  \ref{ass:rdcm_hf} and \ref{ass:rdcm_ident}, which ensure that
\(M_\infty(\theta^-)\) admits a unique maximizer \(\theta^-_0\), then any measurable maximizer
$\hat\theta^-_n\in\arg\underset{\theta^-\in\Theta^-}{\max}M_n(\theta^-)$ satisfies
\[
\hat\theta^-_n\overset{\mathbb P}{\longrightarrow}\theta^-_0.
\]
\end{mycorollary}

\begin{myproof}
See Appendix~\ref{proof:rdcm_pre_tau}.
\end{myproof}

\subsection{SRDCM: fixed-horizon identification of the diffusion blocks}
\label{subsec:srdcm_highfreq}

% We now move from the empirical daily-grid HMM used for filtering to a fixed-horizon continuous-time embedding of the same switching specification. For this reason, we consider below a piecewise constant version $\left\{S^{\Delta_n}_{t}\right\}_{t\in\left[t_0 T_1-\delta\right]}$\footnote{We recall that $T_1$ is the earliest perceived Transition Time} 
% of the Hidden Markov model in \ref{HMM_Switching} and its natural limit \textit{Continuous Time Markov Chain} (CTMC) process $\left\{S_{t}\right\}_{t\in\left[t_0 T_1-\delta\right]}$. 
% As done in Section \ref{subsec:rdcm_highfreq}, we construct an $M_n$ - contrast function that is the natural normalized version of the \textit{direct}-loglikelihood of the SRDCM in Definition \ref{Def_SRDCM_MODEL}. This new contrast $M_{n}$ 
% is the generalization of that in \eqref{eq:Mn_rdcm_main} for a SRDCM context. We show, under a set of assumption that the consistency results in Section \ref{subsec:rdcm_highfreq} can extended for  $j$-regime "visible" diffusion blocks. Based on We conclude with the introduction of a criterion for the identification of the perceived Transition Dates $\left\{T_{j}\right\}$ and we suggest an operation monitoring rule that distingues when the fitting strategy  or alternatively the consistency criterion presented here, are most appropriate approach for measuring the perception to Transition of the economic agents. 

We now move from the empirical daily-grid HMM used for filtering to a fixed-horizon continuous-time embedding of the same switching specification. To this end, we consider the piecewise-constant process\footnote{Here \(T_1\) denotes the earliest perceived transition date.} \(\{S_t^{\Delta_n}\}_{t\in[t_0,T_1-\delta]}\) associated with the hidden Markov model in \eqref{HMM_Switching}, together with its natural continuous-time limit \(\{S_t\}_{t\in[t_0,T_1-\delta]}\), modeled as a Continuous-Time Markov Chain (CTMC).
As in Section~\ref{subsec:rdcm_highfreq}, we introduce an \(M_n\)-contrast as the normalized version of the direct log-likelihood of the SRDCM in Definition~\ref{Def_SRDCM_MODEL}. This extends \eqref{eq:Mn_rdcm_main} to the switching framework. Under suitable assumptions, we show that the consistency results of Section~\ref{subsec:rdcm_highfreq} extend to the visible diffusion blocks in a finite number of regimes. 

Without loss of generality, we consider only the two-regime case and write the gap condition as \(l<T_1-\delta\), with \(T_1<T_2\). For each \(n\), we consider the grid in \eqref{eq:grid_hf_rdcm} and define a two-state piecewise-constant Markov chain \(\{S_{t_i}^{\Delta_n}\}_{i=0}^{n-1}\) with transition matrix
\begin{equation}\label{TransitionProbPiecewiseConstant}
\mathbb{P}^{\Delta_n}=e^{Q\Delta_n}
\qquad \text{where} \qquad
Q = \begin{pmatrix}
q_{11} & q_{12} \\
q_{21} & q_{22}
\end{pmatrix}, \quad q_{ij} > 0.
\end{equation}
Specifically, for our purposes, this process is extended to
\begin{equation}
\{S_t^{\Delta_n}\}_{t\in[t_0,T_1-\delta]}, \qquad
S_t^{\Delta_n}:=S_{t_{i-1}}^{\Delta_n},\qquad t\in[t_{i-1},t_i).
\label{eq:piecewise_switch_main}
\end{equation}
On the same interval, we define its homogeneous CTMC limit
\begin{equation}
\{S_t\}_{t\in[t_0,T_1-\delta]},
\label{eq:switch_main}
\end{equation}
with the same generator \(Q\) used in \eqref{TransitionProbPiecewiseConstant}. We stress that, while the process in \eqref{eq:piecewise_switch_main} can change regime only at the grid points in \eqref{eq:grid_hf_rdcm}, the process in \eqref{eq:switch_main} can change regime finitely many times over \([t_0,T_1-\delta]\).\footnote{We do not develop this convergence here, since it is classical and not needed for the arguments below.} Throughout this subsection, the switching process is assumed to be exogenous, namely independent of the Brownian motion driving \(X_t\).

The first step is to properly define the set of visible parameters $\Theta$ for the 2-regime SRDCM model, which reads:
\begin{equation}
    \Theta=\Theta_1\times \Theta_2
\label{SRDCM_vector_space_visible}
\end{equation}
% where $\Theta_j$ is the usual set of diffusion parameters $j$-RDCM. 
where $\Theta_j$ is the usual set of diffusion parameters in a RDCM model for a candidate transition date $T_j$. 
For each regime, we assume that a structural time point \(\tau_j\) is given and that the corresponding decomposition \(\theta_j=(\theta_j^-,\theta_j^+)\) holds.
In the following, we consider fixed $b_1, b_2$,
% drift parameters $\left\{b_j\right\}_{j=1}^2$ 
and $Q$.\footnote{We recall that the consistency of drift parameters and $Q$ requires that the final observation $l$ tends to $+\infty$, which is not structurally our case, as discussed previously.}
 We denote by \(\theta_0=(\theta_{0,1},\theta_{0,2})\) the vector of true diffusion parameters, assumed to be an interior point of \(\Theta\). Finally, for each regime \(j\in\{1,2\}\), we define:
\begin{equation}
\Psi_j(t,\theta_j)
:=
\log\!\big(g_j^2(T_j-t,\theta_j)\big)
+
\frac{g_j^2(T_j-t,\theta_{0,j})}{g_j^2(T_j-t,\theta_j)},
\qquad t\in\left[t_0,T_1-\delta\right].
\label{eq:Psi_switch_main}
\end{equation}
The restriction to \([t_0,T_1-\delta]\) is motivated by the \emph{sequential model updating strategy} introduced in Section~\ref{SectioSRCDM}. For the sake of brevity, all definitions below are formulated on the observation interval \([t_0,l]\), with \(l<T_1-\delta\); the corresponding extensions to any admissible subinterval of \([t_0,T_1-\delta]\) are immediate.

Let $D([t_0,l],\{1,2\})$ be the Skorokhod space and let $S_f([t_0,l])\subset D([t_0,l],\{1,2\})$ be the set of c\`adl\`ag paths with a finite number of regime changes.
% finitely many jumps. 
For a path \(s\in S_f([t_0,l])\), let \(N(s)\) denote the number of switches and write
\begin{equation}
    t_0=\kappa_0(s)<\kappa_1(s)<\cdots<\kappa_{N(s)}(s)<\kappa_{N(s)+1}(s)=l
\end{equation}
for its switch times $\left\{\kappa_{r}\left(s\right)\right\}_{r=1}^{N\left(s\right)}$. On each interval \([\kappa_r(s),\kappa_{r+1}(s))\), the path is constant, and its value is denoted by \(j_r(s)\in\{1,2\}\). Given the grid \eqref{eq:grid_hf_rdcm}, define the left-endpoint approximation\footnote{Both \(s_t^{\Delta_n}\) and \(s_t\) are viewed as path-valued functions on \([t_0,l]\) taking values in \(\{1,2\}\).}
\begin{equation}
s_t^{\Delta_n}:=s(t_{i-1}),\qquad t\in[t_{i-1},t_i).
\label{eq:left_endpoint_switch_main}
\end{equation}
For a fixed discretized path $s^{\Delta_n}$,  we introduce the rescaled log-Gaussian ``bridge'' density $\log \tilde f_{t_i\mid t_{i-1},j}\left(\theta\right)$ given the $j$-regime that reads:
\begin{equation}
\log \tilde f_{t_i\mid t_{i-1},j}\left(\theta\right)= \tilde{R}_{i,j}\left(\theta\right)-\frac{1}{2} \log\left(\frac{v_{t_i\mid t_{i-1},j}\left(\theta\right)}{\Delta_n}\right)-\frac{1}{2}\frac{\Big(X_{t_i}-m_{t_{i}\mid t_{i-1},j}\left(X_{t_{i-1}}\right)\Big)^2}{v_{t_i\mid t_{i-1},j}\left(\theta\right)}
\label{eq:rescaled_bridge_dens_my0}
\end{equation}
where $\tilde{R}_{i,j}\left(\theta\right):=\log \frac{f_{T_j\mid t_i}\left(0;\theta_j\right)}{f_{T_j\mid t_{i-1}}\left(0;\theta_j\right)}$ while $v_{t_i\mid t_{i-1},j}\left(\theta\right)$  and $m_{t_{i}\mid t_{i-1},j}\left(X_{t_{i-1}}\right)$ are the regime-wise versions of  \eqref{eq:var} and \eqref{eq:mean}.
Using \eqref{eq:rescaled_bridge_dens_my0}, we define the conditional path-wise contrast $M_n(\theta\mid s^{\Delta_n})$:\footnote{When necessary, we use $M_n(\theta\mid s):=M_n(\theta\mid s^{\Delta_n})$ to remark that the contrast is constructed on the left-end point approximation of $s$}
\begin{equation}
M_n(\theta\mid s^{\Delta_n})
:=
\frac{1}{n}\sum_{i=1}^n
\log
\tilde f_{t_i\mid t_{i-1},s^{\Delta_n}_{t_{i-1}}}
\big(X_{t_i}\mid X_{t_{i-1}};\theta_{s^{\Delta_n}_{t_{i-1}}}\big).
\label{eq:conditional_contrast_switch_main}
\end{equation}
For $s\in S_f([t_0,l])$, define the path-wise limit contrast
\begin{equation}
M_\infty(\theta\mid s)
:=
-\frac{1}{2(l-t_0)}
\sum_{r=0}^{N(s)}
\int_{\kappa_r(s)}^{\kappa_{r+1}(s)}
\Psi_{j_r(s)}(u,\theta_{j_r(s)})\,du.
\label{eq:Uinfty_path_switch_main}
\end{equation}
Using the log-sum-exp approach, the observed contrast is
\begin{equation}
M_n(\theta)
:=
\frac{1}{n}\log\int_{S_f([t_0,l])}
\exp\!\big(n\,M_n(\theta\mid s^{\Delta_n})\big)\,\Pi_Q(ds),
\label{eq:observed_switch_contrast_main}
\end{equation}
where $\Pi_Q$ is the law of the continuous-time Markov chain $S$ on $S_f([t_0,l])$.
Finally, define
\begin{equation}
M_\infty(\theta):=
\sup_{s\in S_f([t_0,l])} M_\infty(\theta\mid s).
\label{eq:Minfty_switch_main}
\end{equation}

% \begin{myassumption}[Switching high-frequency regularity and identifiability]
\begin{myassumption}
\label{ass:srdcm_hf}
The following conditions hold:
\begin{enumerate}
    \item the latent regime process $\{S_t\}_{t\in[t_0,l]}$ is a homogeneous continuous-time
    Markov chain on $\{1,2\}$, independent of the Brownian motion, with generator
    $Q=(q_{hk})_{h,k=1,2}$ satisfying
    \[
    0<\underline q \le q_{hk}\le \bar q<\infty,\qquad h\neq k;
    \]
     \item for each regime $j\in\{1,2\}$, Assumption~\ref{ass:rdcm_hf} holds regime-wise, meaning that all requirements are satisfied by $g_1$ and $g_2$ uniformly in $j$ (i.e., with the same underlying constants);
    % \item \textcolor{red}{for each regime $j\in\{1,2\}$, Assumption~\ref{ass:rdcm_hf} holds regime-wise, meaning that all requirements are satisfied by $g_1$ and $g_2$ with constants independent of $j$;}
    \item the observed limit contrast $M_\infty$ admits a unique maximizer on the relevant
    visible block.
\end{enumerate}
\end{myassumption}

\begin{mytheorem}
\label{thm:srdcm_uniform}
Under Assumptions~\ref{ass:srdcm_hf},
\[
\sup_{\theta\in\Theta^{\times 2}}
|M_n(\theta)-M_\infty(\theta)|
\xrightarrow{\mathbb P}0,
\qquad n\to\infty.
\]
\end{mytheorem}

\begin{myproof}
See Appendix~\ref{proof:srdcm_uniform}.
\end{myproof}

The following corollary shows that, once both post-$\tau_j$ regions are observed, the full diffusion blocks attached
to $T_1$ and $T_2$ are identified.

% \begin{mycorollary}[Consistency of the diffusion-block estimator in the post-$\tau$ case]
\begin{mycorollary}
\label{cor:srdcm_post_tau}
Let
\[
\hat\theta_n\in\arg\max_{\theta\in\Theta^{\times 2}} M_n(\theta).
\]
Under Assumptions~\ref{ass:srdcm_hf}, if
$l>\max\{\tau_1,\tau_2\}$, then
$\hat\theta_n \stackrel{\mathbb P}{\longrightarrow}\theta_0.$
\end{mycorollary}
\begin{myproof}
See Appendix~\ref{proof:srdcm_post_tau}.
\end{myproof}

% \begin{mycorollary}[Partial consistency in the pre-$\tau_j$ case]
\begin{mycorollary}
\label{cor:srdcm_pre_tau}
Assume that for regime $j\in\{1,2\}$ the observation horizon does not overlap the
post-$\tau_j$ region, namely $l\le \tau_j.$
% \[
% l\le \tau_j.
% \]
Then the regime-$j$ contribution to the limiting contrast is flat in $\theta_j^+$ and depends
only on the visible block $\theta_j^-$. Consequently, if the reduced limit contrast admits a
unique maximizer in $\theta_j^-$, then the corresponding estimator is consistent for
$\theta_{0,j}^-$, whereas $\theta_j^+$ is not identified.
\end{mycorollary}
\begin{myproof}
See Appendix~\ref{proof:srdcm_pre_tau}.
\end{myproof}

The previous results provide consistent estimation of the visible regime-wise diffusion parameters connected to the deadlines. More precisely, in regime \(j\), the diffusion block
\[
\theta_j=(\theta_j^-,\theta_j^+)
\]
enters the coefficient \(g_j(T_j-t,\theta_j)\), so that the post-\(\tau_j\)
component \(\theta_j^+\) carries the information on the convergence profile
associated with the candidate deadline \(T_j\).

\section{Conclusion}\label{concl}
We introduce a market-implied Time to Transition measure and develop an ad-hoc inference framework. The RDCM provides a tractable benchmark through an exact Gaussian bridge likelihood, while the SRDCM captures reallocations of perceived probability mass across competing deadlines on the observation lattice. On fixed
daily grids, the switching model is used as a filtering device for monitoring perceived transition timing. The empirical analysis on German twin bonds illustrates how the framework can be implemented in a clean market setting where
the informational content of the green label can be isolated.

Under fixed-horizon infill schemes, we show that the regime-wise diffusion parameters can be consistently estimated on their visible part. Full or partial consistency depends on whether the observation window reaches the region where
the post-\(\tau\) parameters are active. Thus, the paper provides a continuous-time framework for market-implied transition timing and an identification result for
the diffusion parameters carrying information about competing perceived deadlines.

% \color{red} VERSIONE FINALE 1. 

% These results provide a foundation for a broader theory of market-implied transition timing. A natural next step is the construction of a global penalized detection criterion over admissible belief paths, integrating pathwise temporal decoding rules with an economic switching cost. This approach reflects the fact that revisions of perceived transition dates are not frictionless: they correspond to costly and partially irreversible decisions—such as building renovation, technological adoption, or long-term investment adjustments. Future research can regularize the maximization problem or incorporating augmented switching specifications with controlled belief revisions, ruling out excessive switching and focusing on economically meaningful paths. Such decision-theoretic monitoring procedures will be essential for assessing whether market perception remains effectively anchored to the nearest policy-relevant transition deadline, providing a robust bridge between stochastic inference and climate policy evaluation.

% \color{blue} VERSIONE FINALE 2 

These results provide a foundation for a broader theory of market-implied transition timing. Natural extensions include pathwise temporal decoding rules, augmented switching specifications with controlled belief revisions, and
decision-theoretic monitoring procedures for assessing whether market perception remains anchored to the nearest policy-relevant transition deadline.

\newpage

% \begin{figure}[h]
% \centering
% \includegraphics[width=6cm]{example-image-rectangle}
% \caption{Although we encourage authors to send us the highest-quality figures possible, for peer-review purposes we can accept a wide variety of formats, sizes, and resolutions. Legends should be concise but comprehensive -- the figure and its legend must be understandable without reference to the text. Include definitions of any symbols used and define/explain all abbreviations and units of measurement.}
% \end{figure}

% \begin{description}
% \item[\texttt{blind}] Anonymise all author, affiliation, correspondence
%        and funding information.

% \item[\texttt{lineno}] Adds line numbers.

% \item[\texttt{serif}] Sets the body font to be serif.

% \item[\texttt{twocolumn}] Sets the body text in two-column layout.

% \item[\texttt{num-refs}] Uses the numerical Vancouver bibliography style; see section \ref{sec:bibstyles}.

% \item[\texttt{alpha-refs}] Uses the author-year RSS bibliography style; see section \ref{sec:bibstyles}.
% \end{description}

% \begin{epigraph}{Albert Einstein}
% Anyone who has never made a mistake has never tried anything new.
% \end{epigraph}

\section*{Acknowledgments}
This work was supported by JST CREST Grant Number JPMJCR2115, Japan
and by the European Union - NextGeneration EU PRIN2022 project ``The effects of climate change in the evaluation of financial instruments'' financed by the `Ministero dell’Università e della Ricerca'
with grant number 20225PC98R, CUP Codes: H53D23002200006 and G53D25001960006.
The authors used ChatGPT (OpenAI) solely for language editing, improvement of
English fluency, and minor stylistic checks. All scientific content,
mathematical arguments, empirical analysis, and conclusions were developed,
reviewed, and verified by the authors, who take full responsibility for the
manuscript.

\section*{Conflict of interest}
The authors declare that they have no conflict of interest.

% \section*{Supporting Information}

% Supporting information is information that is not essential to the article, but provides greater depth and background. It is hosted online and appears without editing or typesetting. It may include tables, figures, videos, datasets, etc. More information can be found in the journal's author guidelines or at \url{http://www.wileyauthors.com/suppinfoFAQs}. Note: if data, scripts, or other artefacts used to generate the analyses presented in the paper are available via a publicly available data repository, authors should include a reference to the location of the material within their paper.

% \printendnotes

% Submissions are not required to reflect the precise reference formatting of the journal (use of italics, bold etc.), however it is important that all key elements of each reference are included.
%\bibliographystyle{WileyNJD-APA}
%\bibliographystyle{apalike}
\bibliography{mybib}

\appendix
\section{Proofs for Section~\ref{newmodel}}
\subsection{Proofs of Subsections \ref{Estim_RDCM_SUB} and \ref{subsubSect}}
\begin{myproof}[\textbf{Proof of Proposition~\ref{prop:bridge_mle} }]\label{prop:bridge_mle_Appendix}
For each \(i=1,\dots,n\), Bayes' rule and the Markov property give
\[
f_{t_i\mid t_{i-1},T}(x_{t_i};\psi)
=
\frac{
f_{T\mid t_i}(0;\psi)\,f_{t_i\mid t_{i-1}}(x_{t_i};\psi)
}{
f_{T\mid t_{i-1}}(0;\psi)
}.
\]
Taking logarithms and summing over \(i\) yields \eqref{eq:loglikSimpl}. Moreover, by Assumption \ref{ass:mle}, the maps \(\psi\mapsto m_{t_i\mid t_{i-1}}(\psi)\) and \(\psi\mapsto v_{t_i\mid t_{i-1}}(\psi)\) are continuous for each \(i=1,\dots,n\). Hence, by \eqref{k}, \eqref{sigma0}, and Lemma \ref{Factor}, the bridge mean \(k_{t_i\mid t_{i-1}}(\psi)\) and variance \(\sigma^2_{t_i\mid t_{i-1}}(\psi)\) are continuous on \(\Psi\), with \(\sigma^2_{t_i\mid t_{i-1}}(\psi)>0\). Therefore, each Gaussian term in \(\ell_n(\psi)\) depends continuously on \(\psi\), and so does \(\ell_n\). Since \(\Psi\) is compact, the existence of a maximizer follows from Weierstrass' theorem.
\end{myproof}

% \section{Proofs for Section \ref{subsubSect}}

\begin{myproof}[Proof of Lemma~\ref{lem:ident_prob} ]\label{Appendix_subsubSect}
When $\tau > t_n$, the identification of $b^{+}$ and $\theta^{+}$ relies only on the term $\ln f_{T\mid t_n}(0; \psi) - \ln f_{T\mid t_0}(0; \psi)$ in \eqref{eq:loglikSimpl}, as the decreasing behavior of $\phi(T-t,b)$ and $g(T-t,\theta)$ is active only for $t > \tau$. Differentiating \eqref{eq:loglikSimpl} with respect to $b^{+}$ yields: \begin{equation*} \frac{\partial \ell_n}{\partial b^{+}} = -\left(\frac{0-m_{T\mid t_n}}{v_{T\mid t_n}}\right)\frac{\partial m_{T \mid t_{n}}}{\partial b^{+}} + \left(\frac{0-m_{T\mid t_0}}{v_{T\mid t_0}}\right)\frac{\partial m_{T \mid t_{0}}}{\partial b^{+}}. \end{equation*} Since $\partial m_{T \mid t_{n}}/{\partial b^{+}} = \partial m_{T \mid t_{0}}/\partial b^{+} \neq 0$, we obtain the stationary condition: \begin{equation} \label{eq:identtheta12} m_{T\mid t_n} = R m_{T\mid t_0}, \quad \text{with} \quad R = \frac{v_{T\mid t_n}}{v_{T\mid t_0}} \in (0,1). \end{equation} Similarly, differentiating with respect to $\theta^{+}$ and noting that $\frac{\partial v_{T \mid t_{0}}}{\partial \theta^{+}} = \frac{\partial v_{T \mid t_{n}}}{\partial \theta^{+}} \neq 0$, the condition reduces to: \begin{equation*} \frac{m^2_{T|t_n}-v_{T|t_n}}{v^2_{T|t_n}} = \frac{m^2_{T|t_0}-v_{T|t_0}}{v^2_{T|t_0}}. \end{equation*} Combining this with \eqref{eq:identtheta12} leads to $v^2_{T|t_n} = v^2_{T|t_0}$, which contradicts the integral representation of the variance, since $v_{T|t}$ must be strictly decreasing in $t$. 
\end{myproof}

\subsection{Proofs of Subsections \ref{EMSRDCM}}

\begin{myproof}[\textbf{Proof of Proposition \ref{prop:FB_SRDMC}}]\label{app:proof_prop32}
We prove the forward recursion by induction on $i$.

For $i=1$, by definition $\alpha_1(j)$ is the joint density-probability of
$\left(S^{\bar\Delta}_{0}=j,X_{t_1}=x_{t_1}\right)$ given $X_{t_0}=x_{t_0}$.
Since the event $S^{\bar\Delta}_{0}=j$ selects the regime-$j$ bridge law on
$(t_0,t_1]$, we immediately obtain
\[
\alpha_1(j)=\pi_{0,j}\,f_j(x_{t_1}\mid x_{t_0}).
\]
Now let $i\geq 2$. We consider the joint mixed density $\alpha_i(j)$ defined as
\begin{equation}
\alpha_i(j):=\mathbb{P}\left(x_{t_1},\ldots,x_{t_{i-1}},S^{\bar\Delta}_{i-1}=j,x_{t_i}\mid x_{t_0}\right),
\label{eq:alpha_step_0}
\end{equation}
given $x_{t_0}$\footnote{In \eqref{eq:alpha_step_0} we omit the dependence on the r.v.'s $\left\{X_{t_h}\right\}_{h=0}^{i}$ and we report explicitly only their observations since the meaning is clear from the usual meaning of the probability measure.}. 

Denoting with $\mathcal{A}_{i}(j):=\left\{\omega: X_{t_1}\left(\omega\right)=x_{t_1};\ldots;X_{t_{i-1}}\left(\omega\right)=x_{t_{i-1}}; S^{\bar{\Delta}}_{i-1}=j;X_{t_{i}}\left(\omega\right)=x_{t_i}\right\}$, we have the following set representation:
\begin{equation}
\mathcal{A}_{i}(j)=\bigcup_{h\in\mathcal{S}}\Big(\mathcal{A}_{i-1}(h)\cap\left\{\omega: S^{\bar{\Delta}}_{i-1}=j;X_{t_{i}}\left(\omega\right)=x_{t_i}\right\}\Big).
\label{eq:set_rapresentation}
\end{equation}
The sets $\left\{\mathcal{A}_{i-1}(h)\right\}_{h\in\mathcal{S}}$ are disjoint and, applying the Bayes' rule and assuming \eqref{eq:alpha_step_0} holds at $t_{i-1}$ for all $h\in\mathcal{S}$, we get:
\begin{eqnarray}
\alpha_i(j)&:=&\mathbb{P}\Big(\mathcal{A}_{i}(j) \Big| x_{t_0}\Big)\nonumber\\
&=&\sum_{h \in \mathcal{S}}\Big[\mathbb{P}\Big(\mathcal{A}_{i-1}(h)\Big| x_{t_0}\Big)\mathbb{P}\Big(\left\{\omega: S^{\bar{\Delta}}_{i-1}=j;X_{t_{i}}\left(\omega\right)=x_{t_i}\right\}\Big|\mathcal{A}_{i-1}(h), x_{t_0}\Big)\Big]\nonumber\\
&=&\sum_{h \in \mathcal{S}}\Big[\alpha_{i-1}(h)p_{hj}\,f_j(x_{t_i}\mid x_{t_{i-1}}) \Big]
=
f_j(x_{t_i}\mid x_{t_{i-1}})
\sum_{h=1}^m\alpha_{i-1}(h)\,p_{hj},\nonumber\\
\label{final-res}
\end{eqnarray}
which gives \eqref{eq:alpha_recursion_our}. 
For the backward recursion, the terminal condition $\beta_n(j)=1$ is immediate.
Then, partitioning the future according to the next latent state $S^{\bar\Delta}_{i}=h$ and using the
same argument as above, one obtains the recursion in  \eqref{eq:beta_final} with $\beta_{n}(j)=1$ and $\forall j \in\mathcal{S}$.

% For the backward recursion, the terminal condition $\beta_n(j)=1$ is immediate.
% Then, conditioning on the next latent state $S^{\bar\Delta}_{i}=h$ and using the
% same argument as above, one obtains
% \[
% \beta_i(j)
% =
% \sum_{h\in\mathcal{S}} p_{jh}\,f_h(x_{t_{i+1}}\mid x_{t_i})\,\beta_{i+1}(h),
% \qquad i=n-1,\ldots,1.
% \]
% This proves the result.
\end{myproof}

\section{Proofs for Section~\ref{sec:highfreq}}
\label{app:proofs_highfreq}

\subsection{RDCM proofs}
\label{app:proofs_rdcm}
We first collect three elementary expansions for the RDCM contrast. Throughout this appendix, when no ambiguity arises, the dependence of stochastic quantities on the observed path \(\{X_t\}\) and on the driving Brownian motion \(\{W_t\}\) is kept implicit. In particular, all stochastic order symbols are understood with respect to the underlying probability measure.

\begin{mylemma}
\label{lem:variance_expansion_rdcm}
Under Assumption~\ref{ass:rdcm_hf}, there exists $C_v<\infty$ such that
\[
\sup_{\theta\in\Theta}
\left|
\frac{v_{t_i\mid t_{i-1}}(\theta)}{\Delta_n}-g^2(T-t_{i-1},\theta)
\right|
\le C_v\Delta_n,
\qquad i=1,\dots,n.
\]
In particular,
\[
v_{t_i\mid t_{i-1}}(\theta)=g^2(T-t_{i-1},\theta)\Delta_n+O(\Delta_n^2)
\]
uniformly in $(i,\theta)$.
\end{mylemma}

\begin{myproof}
Write
\[
\frac{v_{t_i\mid t_{i-1}}(\theta)}{\Delta_n}-g^2(T-t_{i-1},\theta)
=
\frac{1}{\Delta_n}\int_{t_{i-1}}^{t_i}
\Big(\Gamma_i(u)g^2(T-u,\theta)-g^2(T-t_{i-1},\theta)\Big)\,du.
\]
Since \(\Gamma_i(u)=e^{-2a(t_i-u)}\) with \(a\) fixed, one has
\[
\sup_{i,u}|\Gamma_i(u)-1|\le C\Delta_n.
\]
Moreover, as \(g\) is continuous on the compact set \([t_0,T-\delta]\times\Theta\) and uniformly Lipschitz in time, the map \((t,\theta)\mapsto g^2(T-t,\theta)\) is uniformly Lipschitz on the same set. Hence
\[
\sup_{\theta\in\Theta}\sup_{u\in[t_{i-1},t_i]}
\big|g^2(T-u,\theta)-g^2(T-t_{i-1},\theta)\big|
\le C\Delta_n.
\]
Therefore,
\[
\sup_{\theta\in\Theta}
\left|
\frac{v_{t_i\mid t_{i-1}}(\theta)}{\Delta_n}-g^2(T-t_{i-1},\theta)
\right|
\le C\Delta_n,
\]
which proves the claim.
\end{myproof}
\begin{mylemma}
\label{lem:prediction_error_rdcm}
Under Assumption~\ref{ass:rdcm_hf}, for every $i=1,\dots,n$,
\[
X_{t_i}-m_{t_i\mid t_{i-1}}\left(X_{t_{i-1}}\right)=g(T-t_{i-1},\theta_0)\Delta W_i+r_{i,n},
\qquad \Delta W_i:=W_{t_i}-W_{t_{i-1}},
\]
where the remainder satisfies
\[
\sup_{1\le i\le n}\|r_{i,n}\|_{L^2}\le C_r\Delta_n.
\]
Hence, the one-step prediction error is the Brownian term of order $O_p(\sqrt{\Delta_n})$
plus a uniform remainder of smaller order $O_p(\Delta_n)$.
\end{mylemma}

\begin{myproof}
The exact linear representation of the innovation gives
\[
X_{t_i}-m_{t_{i}\mid t_{i-1}}\left(X_{t_{i-1}}\right)
=
\int_{t_{i-1}}^{t_i}\Gamma_i^{1/2}(u)\,g(T-u,\theta_0)\,dW_u.
\]
Add and subtract $g(T-t_{i-1},\theta_0)$ inside the stochastic integral:
\[
r_{i,n}
:=
\int_{t_{i-1}}^{t_i}
\Big(
\Gamma_i^{1/2}(u)\,g(T-u,\theta_0)-g(T-t_{i-1},\theta_0)
\Big)\,dW_u.
\]
By It\^o isometry, the boundedness of $\Gamma_i^{1/2}(u)-1$, and the Lipschitz continuity
of $g$ in time,
\[
\mathbb E|r_{i,n}|^2
\le
C\int_{t_{i-1}}^{t_i}(u-t_{i-1})^2\,du
\le C\Delta_n^3.
\]
Therefore $\|r_{i,n}\|_{L^2}\le C\Delta_n$.
\end{myproof}

\begin{mylemma}
\label{lem:quadratic_expansion_rdcm}
Let
\[
Z_i:=\frac{\Delta W_i}{\sqrt{\Delta_n}}\sim N(0,1).
\]
Under Assumption~\ref{ass:rdcm_hf},
\[
\frac{\Big(X_{t_i}-m_{t_{i}\mid t_{i-1}}\left(X_{t_{i-1}}\right)\Big)^2}{v_{t_{i}\mid t_{i-1}}(\theta)}
=
\frac{g^2(T-t_{i-1},\theta_0)}{g^2(T-t_{i-1},\theta)}\,Z_i^2
+\rho_{i,n}(\theta),
\]
and the remainder satisfies
\[
\sup_{\theta\in\Theta}
\left|
\frac{1}{n}\sum_{i=1}^n \rho_{i,n}(\theta)
\right|
\stackrel{\mathbb P}{\longrightarrow}0.
\]
\end{mylemma}

\begin{myproof}
Write
\[
X_{t_i}-m_{t_{i}\mid t_{i-1}}\left(X_{t_{i-1}}\right)=g_i^0\sqrt{\Delta_n}\,Z_i+r_{i,n},
\qquad g_i^0:=g(T-t_{i-1},\theta_0),
\]
and
\[
v_{t_{i}\mid t_{i-1}}(\theta)=g_i^2(\theta)\Delta_n+\delta_{i,n}(\theta),
\qquad g_i(\theta):=g(T-t_{i-1},\theta),
\]
where, by Lemma~\ref{lem:variance_expansion_rdcm},
\[
\sup_{\theta\in\Theta}|\delta_{i,n}(\theta)|\le C\Delta_n^2.
\]
Then
\[
\Big(X_{t_i}-m_{t_{i}\mid t_{i-1}}\left(X_{t_{i-1}}\right)\Big)^2
=
g_i^{0\,2}\Delta_n Z_i^2
+
2g_i^0\sqrt{\Delta_n}\,Z_i r_{i,n}
+
r_{i,n}^2.
\]
Using the lower bound $g_i^2(\theta)\ge c_g$, we may divide by $v_i(\theta)$ and obtain
the announced expansion. The average remainder vanishes because
\[
\mathbb E\!\left[
\left|
\frac{2g_i^0\sqrt{\Delta_n}Z_i r_{i,n}}{v_{t_{i}\mid t_{i-1}}(\theta)}
\right|
\right]
\le C\sqrt{\Delta_n},
\qquad
\mathbb E\!\left[
\left|
\frac{r_{i,n}^2}{v_{t_{i}\mid t_{i-1}}(\theta)}
\right|
\right]
\le C\Delta_n,
\qquad
\sup_{\theta\in\Theta}
\left|
\frac{\delta_{i,n}(\theta)}{v_{t_{i}\mid t_{i-1}}(\theta)}
\right|
\le C\Delta_n.
\]
Averaging over $i$ yields the result.
\end{myproof}

\begin{myproof}[Proof of Proposition~\ref{prop:rdcm_uniform} ]\label{proof:rdcm_uniform}

We split $M_n$ as in \eqref{eq:Mn_rdcm_main}.

\medskip
\noindent

\emph{Step 1: terminal bridge term.}
Since \(l\le T-\delta\), the Gaussian bridge log-density \(\log f_{T\mid t}(0;X_t,\theta)\) is uniformly \(O_p(1)\) on \([t_0,l]\times\Theta\). Indeed, \(\log v_{T\mid t}(\theta)\) is uniformly \(O(1)\), because \(v_{T\mid t}(\theta)\) is continuous on \([t_0,T]\times\Theta\) and, by the gap condition together with continuity and uniform ellipticity up to \(T-\delta\), it is uniformly bounded away from zero. In particular, this lower bound does not degenerate even in the limit case \(t=T-\delta\), as follows from the variance representation in \eqref{eq:var}.  The quadratic kernel is uniformly \(O_p(1)\), since \(m_{T\mid t}(X_t)=O_p(1)\) on \([t_0,l]\) by the continuity of the sample paths and the boundedness of the deterministic integral term in \eqref{eq:mean}, while \(v_{T\mid t}(\theta)=O(1)\) uniformly. Therefore,
\begin{equation}
\frac{1}{n}\sup_{\theta\in\Theta}|R_n(\theta)|=O_p(n^{-1})=o_p(1).
\label{eq:ControlTelescopicTerm}
\end{equation}

\medskip
\noindent
\emph{Step 2: log-variance term.}
By Lemma~\ref{lem:variance_expansion_rdcm},
\[
\frac{v_{t_{i}\mid t_{i-1}}(\theta)}{\Delta_n}
=
g^2(T-t_{i-1},\theta)\big(1+O(\Delta_n)\big)
\]
uniformly over all observation times $t_i\leq l$ and over $\Theta$. Using the uniform ellipticity, we have 
\[
\log\left(\frac{v_{t_{i}\mid t_{i-1}}(\theta)}{\Delta_n}\right)
=
\log\Big(g^2(T-t_{i-1},\theta)\Big)+O(\Delta_n)
\]
where the last term comes from $\log(1+x)=O(x)$. Therefore, considering the grid in \eqref{eq:grid_hf_rdcm}, we obtain:
\begin{equation}
\frac{1}{l-t_0}\sum_{t_i\leq l}\left[\log\left(\frac{v_{t_{i}\mid t_{i-1}}(\theta)}{\Delta_n}\right)-\log\Big(g^2(T-t_{i-1},\theta)\Big)\right]\Delta_n=o(1)
\label{aaaaauffa}
\end{equation}
uniformly on $\Theta$. Since $\left(t,\theta\right)\mapsto\log\Big(g^{2}\left(T-t,\theta\right)\Big)$ is continuous on $\left[t_0,l\right]\times\Theta$, the corresponding left-Riemann sums converge uniformly in $\Theta$ to the integral. Combining this uniform convergence of the left-Riemann sums with \eqref{aaaaauffa}, finally we get: 
\begin{equation}
    \sup_{\theta\in\Theta}
\left|
-\frac{1}{2(l-t_0)}\sum_{t_i\leq l}\log\left(\frac{v_{t_{i}\mid t_{i-1}}(\theta)}{\Delta_n}\right)\Delta_n
+
\frac{1}{2(l-t_0)}\int_{t_0}^{l}\log\!\big(g^2(T-u,\theta)\big)\,du
\right|
\to 0.
\label{eq:stepcontrolkernel}
\end{equation}

\medskip
\noindent
\emph{Step 3: quadratic term.}
By Lemma~\ref{lem:quadratic_expansion_rdcm},
\[
-\frac{1}{2n}\sum_{i=1}^n
\frac{\Big(X_{t_i}-m_{t_i\mid t_{i-1}}(X_{t_{i-1}})\Big)^2}{v_{t_i\mid t_{i-1}}(\theta)}
=
-\frac{1}{2n}\sum_{i=1}^n w_i(\theta)Z_i^2+o_p(1),
\]
uniformly in \(\theta\), where
\[
w_i(\theta):=\frac{g^2(T-t_{i-1},\theta_0)}{g^2(T-t_{i-1},\theta)}.
\]
Using the observation grid in \eqref{eq:grid_hf_rdcm}, we decompose
\[
-\frac{1}{2n}\sum_{i=1}^n w_i(\theta)Z_i^2
=
-\frac{1}{2(l-t_0)}\sum_{t_i\le l} w_i(\theta)\Delta_n
-\frac{1}{2(l-t_0)}\sum_{t_i\le l} w_i(\theta)(Z_i^2-1)\Delta_n.
\]
Since \((t,\theta)\mapsto g^2(T-t,\theta_0)/g^2(T-t,\theta)\) is continuous on
\([t_0,l]\times\Theta\), the corresponding left-Riemann sums converge uniformly in
\(\Theta\), that is,
\[
\sup_{\theta\in\Theta}
\left|
-\frac{1}{2(l-t_0)}\sum_{t_i\le l} w_i(\theta)\Delta_n
+
\frac{1}{2(l-t_0)}
\int_{t_0}^{l}\frac{g^2(T-u,\theta_0)}{g^2(T-u,\theta)}\,du
\right|
\to 0.
\]
It remains to control the centered term. We claim that
\begin{equation}
\sup_{\theta\in\Theta}
\left|
\frac{1}{2(l-t_0)}\sum_{t_i\le l} w_i(\theta)(Z_i^2-1)\Delta_n
\right|
\stackrel{\mathbb P}{\longrightarrow}0.
\label{eq:weighted_lln_appendix}
\end{equation}
Indeed, to prove \eqref{eq:weighted_lln_appendix}, it is enough to apply Chebyshev's inequality and show that the corresponding second moment converges to zero uniformly in \(\Theta\), namely
\begin{equation}
\frac{1}{4(l-t_0)^2}
\sup_{\theta\in\Theta}
\mathbb E_{\theta_0}\!\left[
\left(
\sum_{t_i\le l} w_i(\theta)(Z_i^2-1)\Delta_n
\right)^2
\right]
\longrightarrow 0,
\qquad n\to+\infty.
\label{eq:Control_Estimated_Variance}
\end{equation}
Expanding the square in \eqref{eq:Control_Estimated_Variance}, the mixed terms corresponding to \(t_i\neq t_j\) have zero mean by the independence of the sequence \(\{Z_i\}_{i=1}^n\). Hence,
\[
\sup_{\theta\in\Theta}
\sum_{t_i\le l}
w_i(\theta)^2\,
\mathbb E_{\theta_0}\!\left[(Z_i^2-1)^2\right]
\Delta_n^2
\le
2C(l-t_0)\Delta_n
\longrightarrow 0.
\]
Here \(\{w_i\}_{i=1}^n\) is uniformly bounded on the grid \eqref{eq:grid_hf_rdcm} by ellipticity and continuity of \((t,\theta)\mapsto g(T-t,\theta)\) on the relevant compact set. Therefore \eqref{eq:Control_Estimated_Variance} yields \eqref{eq:weighted_lln_appendix}. Combining this with \eqref{eq:stepcontrolkernel} and \eqref{eq:ControlTelescopicTerm} proves the uniform convergence of \(M_n(\theta)\) to \(M_\infty(\theta)\) in \eqref{eq:ConvergUnPofMn_for RDCM}.
By continuity of \(M_n(\cdot)\) and compactness of \(\Theta\), \(\arg\max_{\theta\in\Theta}M_n(\theta)\neq\varnothing\). The consistency of any measurable selection \(\hat\theta_n\) then follows from \eqref{eq:ConvergUnPofMn_for RDCM}, Assumption~\ref{ass:rdcm_ident}, and the argmax theorem in \cite[Chapter~5]{van1998asymptotic}.
% Thus it is enough to prove that
% \begin{equation}
% \sup_{\theta\in\Theta}
% \left|
% \frac{1}{n}\sum_{i=1}^n w_i(\theta)(Z_i^2-1)
% \right|
% \stackrel{\mathbb P}{\longrightarrow}0.
% \label{eq:weighted_lln_appendix}
% \end{equation}
% Fix $\eta>0$ and let $\{\theta^{(1)},\dots,\theta^{(m)}\}$ be a finite $\eta$-net of $\Theta$.
% By the uniform Lipschitz continuity of $g$ and the lower bound $g^2\ge c_g$, there
% exists $C_w<\infty$ such that
% \[
% |w_i(\theta)-w_i(\vartheta)|\le C_w\|\theta-\vartheta\|_2,
% \qquad \forall i,\ \forall\theta,\vartheta\in\Theta.
% \]
% Therefore
% \[
% \sup_{\theta\in\Theta}
% \left|
% \frac{1}{n}\sum_{i=1}^n w_i(\theta)(Z_i^2-1)
% \right|
% \le
% \max_{1\le k\le m}
% \left|
% \frac{1}{n}\sum_{i=1}^n w_i(\theta^{(k)})(Z_i^2-1)
% \right|
% +
% C_w\eta\cdot \frac{1}{n}\sum_{i=1}^n |Z_i^2-1|.
% \]
% The first term converges to zero in $L^2$ because the net is finite and the $Z_i$ are
% independent with $\mathbb E[Z_i^2-1]=0$ and $\mathrm{Var}(Z_i^2-1)<\infty$.
% The second term can be made arbitrarily small by choosing $\eta$ small and using the
% law of large numbers for $n^{-1}\sum |Z_i^2-1|$. This proves \eqref{eq:weighted_lln_appendix}.

% Consistency of $\hat\theta_n$ then follows from the standard argmax theorem on the
% compact parameter space $\Theta$ and the uniqueness of the maximizer in
% Assumption~\ref{ass:rdcm_ident}.
\end{myproof}

\begin{myproof}[Proof of Corollary~\ref{cor:rdcm_pre_tau}]
\label{proof:rdcm_pre_tau}
If $l\leq\tau$, then on the whole observation interval $[t_0,l]$ the diffusion coefficient
does not depend on $\theta^+$. Hence the same is true for the limiting contrast, and
\[
M_\infty(\theta^-,\theta^+)=M_\infty^-(\theta^-).
\]
The consistency of $\hat\theta_n^-$ follows by applying Proposition~\ref{prop:rdcm_uniform}
to the reduced visible block. No consistency claim can be made for $\hat\theta_n^+$
because the limit contrast is flat in that component.
\end{myproof}

\subsection{Switching proofs}
\label{app:proofs_switching}

% Requires in the preamble:
% \usepackage{graphicx}
% \usepackage{tikz}
% \usetikzlibrary{arrows.meta,calc,positioning}

\begin{figure}[htbp]
\centering
\resizebox{\textwidth}{!}{%
\begin{tikzpicture}[x=1cm,y=1cm,>=Latex]

% -------------------------------------------------
% Bounding box: enough depth so caption stays below
% -------------------------------------------------
\path[use as bounding box] (-0.25,-4.90) rectangle (12.55,4.85);

% -------------------------------------------------
% Colors
% -------------------------------------------------
\definecolor{goodgreen}{RGB}{207,238,205}
\definecolor{goodedge}{RGB}{45,123,47}
\definecolor{badred}{RGB}{246,197,197}
\definecolor{badedge}{RGB}{155,27,27}
\definecolor{jumpblue}{RGB}{31,90,166}

% -------------------------------------------------
% Basic coordinates
% -------------------------------------------------
\def\xmin{0}
\def\xmax{12}

\def\kone{2.5}
\def\ktwo{6.5}
\def\kthree{9.5}

\def\okone{2}
\def\ukone{3}
\def\oktwo{6}
\def\uktwo{7}
\def\okthree{9}
\def\ukthree{10}

% -------------------------------------------------
% Vertical layout: only spacing increased
% -------------------------------------------------
\def\ytitle{4.50}
\def\ysubtitle{4.08}

\def\ystateLine{3.10}
\def\ystateText{3.35}

\def\ykappa{2.45}
\def\ydashTop{2.15}
\def\ydashBottom{-0.25}

\def\ygoodText{1.72}
\def\ygoodBottom{0.82}
\def\ygoodTop{1.20}

\def\ygridTop{0.50}
\def\ygridMid{0.00}
\def\ygridBot{-0.50}
\def\ytickTop{0.48}
\def\ytickBot{-0.48}

\def\ybadText{-0.3}
\def\yticksText{-0.92}

\def\yover{-1.60}
\def\yunder{-2.18}

\def\yboxTop{-2.10}

% -------------------------------------------------
% Title and subtitle
% -------------------------------------------------
\node[anchor=west,font=\bfseries] at (0,\ytitle)
{SRDCM pathwise decomposition on the observation grid};

\node[anchor=west,font=\normalsize] at (0,\ysubtitle)
{\ \  \ Timeline with $N(s)=3$: green = regime-stable cell, 
red = switch cell.};

% -------------------------------------------------
% State path on top
% -------------------------------------------------
\draw[line width=1.2pt,color=blue] (0,\ystateLine) -- (\kone,\ystateLine);
\draw[line width=1.2pt,color=orange] (\kone,\ystateLine) -- (\ktwo,\ystateLine);
\draw[line width=1.2pt,color=green!60!black] (\ktwo,\ystateLine) -- (\kthree,\ystateLine);
\draw[line width=1.2pt,color=red] (\kthree,\ystateLine) -- (12,\ystateLine);

\node[font=\small] at (1.08,\ystateText) {state $j_0(s)$};
\node[font=\small] at (4.20,\ystateText) {state $j_1(s)$};
\node[font=\small] at (7.86,\ystateText) {state $j_2(s)$};
\node[font=\small] at (10.74,\ystateText) {state $j_3(s)$};

% -------------------------------------------------
% Jump times
% -------------------------------------------------
\draw[dashed,line width=0.8pt,color=jumpblue] (\kone,\ydashBottom) -- (\kone,\ydashTop);
\draw[dashed,line width=0.8pt,color=jumpblue] (\ktwo,\ydashBottom) -- (\ktwo,\ydashTop);
\draw[dashed,line width=0.8pt,color=jumpblue] (\kthree,\ydashBottom) -- (\kthree,\ydashTop);

\node[font=\normalsize,text=jumpblue] at (\kone,\ykappa) {$\kappa_1(s)$};
\node[font=\normalsize,text=jumpblue] at (\ktwo,\ykappa) {$\kappa_2(s)$};
\node[font=\normalsize,text=jumpblue] at (\kthree,\ykappa) {$\kappa_3(s)$};

% -------------------------------------------------
% Good-part labels
% -------------------------------------------------
\node[font=\small,text=goodedge,align=center] at (1.0,\ygoodText)
{stable part 0\\$[t_0,\overline{\kappa}_1)$};

\node[font=\small,text=goodedge,align=center] at (4.5,\ygoodText)
{stable part 1\\$[\underline{\kappa}_1,\overline{\kappa}_2)$};

\node[font=\small,text=goodedge,align=center] at (8.0,\ygoodText)
{stable part 2\\$[\underline{\kappa}_2,\overline{\kappa}_3)$};

\node[font=\small,text=goodedge,align=center] at (11.0,\ygoodText)
{stable part 3\\$[\underline{\kappa}_3,\ell]$};

% -------------------------------------------------
% Good-part bars
% -------------------------------------------------
\filldraw[fill=goodgreen,draw=goodedge,line width=0.6pt]
    (0,\ygoodBottom) rectangle (2,\ygoodTop);
\filldraw[fill=goodgreen,draw=goodedge,line width=0.6pt]
    (3,\ygoodBottom) rectangle (6,\ygoodTop);
\filldraw[fill=goodgreen,draw=goodedge,line width=0.6pt]
    (7,\ygoodBottom) rectangle (9,\ygoodTop);
\filldraw[fill=goodgreen,draw=goodedge,line width=0.6pt]
    (10,\ygoodBottom) rectangle (12,\ygoodTop);

% -------------------------------------------------
% Observation grid
% -------------------------------------------------
\foreach \i in {0,...,11}{
    \pgfmathtruncatemacro{\j}{\i+1}
    \def\fillcolor{white}
    \def\drawcolor{black}
    \def\fillopacity{1}
    \ifnum\i=2 \def\fillcolor{badred} \def\drawcolor{badedge} \def\fillopacity{0.35}  \fi
    \ifnum\i=6 \def\fillcolor{badred} \def\drawcolor{badedge} \def\fillopacity{0.35} \fi
    \ifnum\i=9 \def\fillcolor{badred} \def\drawcolor{badedge} \def\fillopacity{0.35} \fi
     \filldraw[
        fill=\fillcolor,
        fill opacity=\fillopacity,
        draw=\drawcolor,
        draw opacity=1,
        line width=0.5pt
    ] (\i,\ygridBot) rectangle (\j,\ygridTop);
}

\draw[line width=1.2pt] (\xmin,\ygridMid) -- (\xmax,\ygridMid);

\foreach \i in {0,...,12}{
    \draw[line width=0.5pt] (\i,\ytickBot) -- (\i,\ytickTop);
}

\fill[jumpblue] (\kone,\ygridMid) circle (1.2pt);
\fill[jumpblue] (\ktwo,\ygridMid) circle (1.2pt);
\fill[jumpblue] (\kthree,\ygridMid) circle (1.2pt);

% -------------------------------------------------
% Bad-cell labels
% -------------------------------------------------
\node[font=\small,text=badedge] at (2.5,\ybadText) {switch};
\node[font=\small,text=badedge] at (6.5,\ybadText) {switch};
\node[font=\small,text=badedge] at (9.5,\ybadText) {switch};

% -------------------------------------------------
% Grid-point labels t_i
% -------------------------------------------------
\node[font=\small] at (0,\yticksText) {$t_0$};
\foreach \i in {1,...,12}{
    \node[font=\small] at (\i,\yticksText) {$t_{\i}$};
}

% -------------------------------------------------
% Overline / underline labels below
% -------------------------------------------------
\node[font=\small,text=badedge] at (\okone,\yover) {$\overline{\kappa}_1$};
\node[font=\small,text=badedge] at (\oktwo,\yover) {$\overline{\kappa}_2$};
\node[font=\small,text=badedge] at (\okthree,\yover) {$\overline{\kappa}_3$};

\node[font=\small,text=badedge] at (\ukone,\yover) {$\underline{\kappa}_1$};
\node[font=\small,text=badedge] at (\uktwo,\yover) {$\underline{\kappa}_2$};
\node[font=\small,text=badedge] at (\ukthree,\yover) {$\underline{\kappa}_3$};

% -------------------------------------------------
% Definition box
% -------------------------------------------------
\node[
    draw=gray!70,
    rounded corners=2pt,
    fill=gray!10,
    align=left,
    font=\normalsize,
    anchor=north west,
    text width=8cm,
    inner sep=6pt
] at (0,\yboxTop)  {%
$\overline{\kappa}_r(s) := \sup\{t_i : t_i \leq \kappa_r(s)\}$\\
$\underline{\kappa}_r(s) := \inf\{t_i : t_i \geq \kappa_r(s)\}$\\
Regime-stable cells on regime segment $r$: $[\underline{\kappa}_r(s),\overline{\kappa}_{r+1}(s))$\\
Switch cells around  switching time  $r$: $[\overline{\kappa}_r(s),\underline{\kappa}_r(s))$
};

\end{tikzpicture}%
}
\caption{Pathwise decomposition of the observation grid for a trajectory \(s\in S_f([t_0,\ell])\) with \(N(s)=3\). Green intervals correspond to regime-constant cells, while red cells contain
regime-switching times and are referred to as switch cells.
% Green intervals correspond to the good part handled regimewise. Specifically within the green cells the SRDCM behaves as a RDCM since the regime does not change. The red cells contain the jump times and generate the bad part.
}
\label{fig:srdcm-good-bad-decomposition}
\end{figure}

Following the decomposition in Figure \ref{fig:srdcm-good-bad-decomposition}, we define the switch-cell set 
\begin{equation}
J_n(s)
:=
\Big\{
i\in\{1,\dots,n\}:
\{\kappa_1(s),\dots,\kappa_{N(s)}(s)\}\cap (t_{i-1},t_i)\neq\emptyset
\Big\},
\label{bad_cell}
\end{equation}
for any $s\in S_f([t_0,l])$, as the set of indices of subintervals in grid \eqref{eq:grid_hf_rdcm} that contain a switch-time $\kappa_{r}(s)$.
The set $G_n(s):=\{1,\dots,n\}\setminus J_n(s)$ is the complement of \eqref{bad_cell} over the observation grid \eqref{eq:grid_hf_rdcm}. 
As discussed in the following lemma, it is possible to identify for any $s\in S_f([t_0,l])$ a positive integer $n_0(s)$ such that, for any $i \in J_n(s)$, the subinterval $\left[t_{i-1}, t_i\right)$ contains only one switch time.

\begin{mylemma}
\label{lem:one_jump_per_cell}
For every fixed path $s\in S_f([t_0,l])$, there exists $n_0(s)<\infty$ such that for all
$n\ge n_0(s)$ each interval $\left[t_{i-1},t_i\right)$ contains at most one switch of $s$.
Consequently,
\[
|J_n(s)|\le N(s),\qquad n\ge n_0(s).
\]
\end{mylemma}

\begin{myproof}
If $N(s)=0$, there is nothing to prove. Otherwise define
\[
d_*(s):=\min_{1\le r\neq h\le N(s)}|\kappa_r(s)-\kappa_h(s)|>0.
\]
Since $\Delta_n\to 0$, there exists $n_0(s)$ such that $\Delta_n<d_*(s)$ for all
$n\ge n_0(s)$. For such $n$, no cell can contain two switches.
\end{myproof}
Let 
\[\overline{\kappa}_r(s) := \sup\{t_i : t_i \leq \kappa_r(s)\}, \qquad \underline{\kappa}_r(s) := \inf\{t_i : t_i \geq \kappa_r(s)\}
\]
We decompose the path-wise limit contrast $M_\infty\left(\theta\mid s\right)$ in \eqref{eq:Uinfty_path_switch_main} in two parts. Specifically, we have: 
\begin{equation}
M_{\infty}\left(\theta \mid s\right)=M^{(0)}_{\infty}\left(\theta \mid s\right)  + M^{(J)}_{\infty}\left(\theta \mid s\right)  
\label{eq:limit_decop_pathwise}
\end{equation}
with $$M^{(0)}_{\infty}\left(\theta \mid s\right):=-\frac{1}{2(l-t_0)}\overset{N(s)}{\underset{r=0}{\sum}} \int_{\underline{\kappa}_{r}(s)}^{\overline{\kappa}_{r+1}(s)}\Psi_{j_r(s)}(u,\theta_{j_r(s)})\,du,$$ $\underline{\kappa}_{r}(s)=0$,  and $M^{(J)}_{\infty}\left(\theta \mid s\right):=-\frac{1}{2(l-t_0)}\overset{N(s)}{\underset{r=1}{\sum}} \int_{\overline{\kappa}_{r}(s)}^{\underline{\kappa}_{r}(s)}\Psi_{j_r(s)}(u,\theta_{j_r(s)})\,du$

Accordingly, we decompose the contrast \eqref{eq:conditional_contrast_switch_main}
\begin{equation}
M_n(\theta\mid s^{\Delta_n})
=
M_n^{(0)}(\theta\mid s^{\Delta_n})
+
M_n^{(J)}(\theta\mid s^{\Delta_n}).
\label{eq:DecomMnPath}
\end{equation}
where $M_n^{(0)}:=\frac{1}{n}\underset{i \in G_n\left(s\right)}{\sum}
\log
\tilde f_{t_i\mid t_{i-1},s^{\Delta_n}_{t_{i-1}}}
\big(X_{t_i}\mid X_{t_{i-1}};\theta_{s^{\Delta_n}_{t_{i-1}}}\big)$ is the contribution of the regime-stable cells and $$M_n^{(J)}:=\frac{1}{n}\underset{i \in J_n\left(s\right)}{\sum}
\log
\tilde f_{t_i\mid t_{i-1},s^{\Delta_n}_{t_{i-1}}}
\big(X_{t_i}\mid X_{t_{i-1}};\theta_{s^{\Delta_n}_{t_{i-1}}}\big)$$ that of the switch cells. Using the structure in \eqref{eq:DecomMnPath} and 
\eqref{eq:limit_decop_pathwise} we can naturally achieve the following path-wise uniform convergence in probability for $M_{n}\left(\theta\mid s^{\Delta_n}\right)$ to $M_{\infty}\left(\theta\mid s^{\Delta_n}\right)$ over the compact set. 
\begin{mylemma}\label{lem:pathwise_uniform_switch}
Under Assumption~\ref{ass:srdcm_hf}, for every fixed path
$s\in S_f([t_0,l])$, let $n_0(s)<\infty$ be as in Lemma~\ref{lem:one_jump_per_cell}.
Then, for all $n\ge n_0(s)$, the following results hold:
\begin{enumerate}
    \item[(i)] There exists a random variable $\Xi_n(s)=O_p(1)$, uniform in
$\theta\in\Theta^{\times 2}$, such that
\begin{equation}
\sup_{\theta\in\Theta^{\times 2}}
\bigl|M_n^{(J)}(\theta\mid s^{\Delta_n})\bigr|
\le
\frac{N(s)}{n}\,\Xi_n(s).
\label{aglioprezzemolo}
\end{equation}
\item[(ii)] $M_{\infty}^{J}\left(\theta\mid s\right)$ converges uniformly to 0 as $n\rightarrow+\infty$ over $\Theta^2$, specifically $\underset{{\theta\in\Theta^{\times 2}}}\sup\big| M_{\infty}^{J}\left(\theta\mid s\right)\big|=O\left(\frac{N(s)}{n}\right)$.
\item[(iii)] $M_n^{(J)}(\theta\mid s^{\Delta_n})$ converges uniformly in probability to $M_{\infty}^{(J)}\left(\theta\mid s\right)$ on the compact set $\Theta^{\times2}$, specifically:
\begin{equation}
\sup_{\theta\in\Theta^{\times 2}}\Big|M_{n}^{\left(J\right)}\left(\theta\mid s^{\Delta_n}\right)-M_{\infty}^{\left(J\right)}\left(\theta\mid s\right)\Big|=O_p\left(\frac{N(s)}{n}\right), \qquad n\rightarrow \infty.
\label{echecavolo}
\end{equation}
\item[(iv)] $M_n\left(\theta\mid s^{\Delta_n}\right)$ and $M_{\infty}\left(\theta\mid s\right)$ satisfy:
\begin{equation}
\sup_{\theta\in\Theta^{\times2}}\Big|M_{n}\left(\theta\mid s^{\Delta_n}\right)-M_{\infty}\left(\theta\mid s\right)\Big|=o_p(1)+O_p\left(\frac{N(s)}{n}\right), \qquad n\rightarrow \infty.    
\label{uniformlyBoundenessPathwise}
\end{equation}
since, for any $s \in S_f([t_0,l])$, $N(s)<\infty$, the result in  \eqref{uniformlyBoundenessPathwise}
implies the following uniform path-wise convergence, that is:
\begin{equation}
\sup_{\theta\in\Theta^{\times2}}\Big|M_{n}\left(\theta\mid s^{\Delta_n}\right)-M_{\infty}\left(\theta\mid s\right)\Big| \overset{\mathbb P}{\longrightarrow}0.
\label{eq:PathwiseUNIFORMCONVERGENCEOFMn}
\end{equation}
\end{enumerate}
\end{mylemma}
\begin{myproof}
\begin{enumerate}
\item[(i)] To obtain \eqref{aglioprezzemolo}, using set $J_n(s)$ in \eqref{bad_cell} we write $M_{n}^{\left(J\right)}\left(\theta\mid s^{\Delta_n}\right)$ as: 
\begin{equation*}
    M_{n}^{\left(J\right)}\left(\theta\mid s^{\Delta_n}\right)=\frac{1}{n}\sum_{i \in J_n\left(s\right)}\log \tilde{f}_{t_i\mid t_{i-1}, s^{\Delta_n}_{t_{i-1}}}\Big(X_{t_i}\mid X_{t_{i-1}};\theta_{s^{\Delta_n}_{t_{i-1}}}\Big)
\end{equation*}
where the log-rescaled density $\log \tilde{f}_{t_i\mid t_{i-1}, s^{\Delta_n}_{t_{i-1}}}\Big(X_{t_i}\mid X_{t_{i-1}};\theta_{s^{\Delta_n}_{t_{i-1}}}\Big)$ is defined in \eqref{eq:rescaled_bridge_dens_my0}. Using the same arguments for the RDCM, the latter quantity $\log \tilde{f}_{t_i\mid t_{i-1}, s^{\Delta_n}_{t_{i-1}}}\Big(X_{t_i}\mid X_{t_{i-1}};\theta_{s^{\Delta_n}_{t_{i-1}}}\Big)$  is $O_p(1)$ uniformly on $\left[t_0, l\right]\times \Theta^{\times2}$ due to Assumption \ref{ass:srdcm_hf} - (2) that upgrades Assumption \ref{ass:rdcm_hf} regime-wise.\footnote{In particular, the uniform boundedness of the log-rescaled density arises from the regime-wise uniform ellipticity condition and the uniform Lipschitz continuity of $g_j$, regime-wise.} Therefore:
\[
\frac{1}{n}\sup_{\theta \in \Theta^{\times2}}\Big|\sum_{i \in J_n(s)}\log \tilde{f}_{t_i\mid t_{i-1}, s^{\Delta_n}_{t_{i-1}}}\Big(X_{t_i}\mid X_{t_{i-1}};\theta_{s^{\Delta_n}_{t_{i-1}}}\Big)\Big|=O_p(1)\frac{\left|J_n(s)\right|}{n}
\] where $\left|J_n(s)\right|$ is the cardinality of \eqref{bad_cell}. Using Lemma \ref{lem:one_jump_per_cell} we have the result in \eqref{aglioprezzemolo}.
\item[(ii)] By the regime-wise version of Assumption~\ref{ass:srdcm_hf},
the function $\Psi_{j_r(s)}(u,\theta_{j_r(s)})$ is uniformly bounded on the
relevant compact set. Moreover, by Lemma~\ref{lem:one_jump_per_cell}, each switch
cell contains only one switch, hence on each switch interval the integrand has at
most one jump discontinuity. Therefore,

\[
\frac{1}{2\left(l-t_0\right)}\sup_{\theta \in \Theta^{\times2}}\Big|\sum_{r=1}^{N(s)}\int^{\underline{\kappa}_{r}(s)}_{\overline{\kappa}_{r}(s)}\Psi_{j_r(s)}(u,\theta_{j_r(s)})\,du\Big|=O\left(\frac{N(s)\Delta_n}{2\left(l-t_0\right)}\right).
\]
Since $\left(l-t_0\right)=n \Delta_n$ and $N(s)<\infty$ for any $s\in S_f\left(\left[t_0,l\right]\right)$ we obtain $M_{\infty}^{\left(J\right)}\left(\theta|s\right)=O\left(\frac{N(s)}{n}\right)$.
\item[(iii)] The result in \eqref{echecavolo} comes directly from the previous two points.
 \item[(iv)] On each regime-stable interval determined by the fixed path $s$, the active regime is constant. Hence, by applying Proposition~\ref{prop:rdcm_uniform} regime-wise, the regime-stable cell contribution is $o_p(1)$. Combining this with point (iii) and using the triangle inequality yields \eqref{uniformlyBoundenessPathwise}. Since $N(s)<\infty$ for every fixed $s\in S_f([t_0,l])$, \eqref{eq:PathwiseUNIFORMCONVERGENCEOFMn} follows.

\end{enumerate}
\end{myproof}
The following Lemma plays the role of a switching-specific global comparison bound for the log-sum-exp aggregation over admissible paths, ensuring that the positive deviation of the observed contrast remains uniformly comparable across paths; compare \cite{Yoshida2022QLA,cheng2026quasi}

\begin{mylemma}
\label{lem:global_positive_comparison}
Under Assumption~6, there exist a deterministic constant \(C_{\star}>0\) and a sequence of nonnegative random variables \(\{r_n\}_{n\ge1}\) such that $r_n=o_P(1),$
% \[
% r_n=o_P(1),
% \]
and, for every \(s\in S_f([t_0,l])\) and every \(n\ge1\),
\begin{equation}
\sup_{\theta\in\Theta^{\times 2}}
\Bigl[
M_n\!\left(\theta \mid s^{\Delta_n}\right)
-
M_\infty\!\left(\theta \mid s\right)
\Bigr]_+
\le
r_n+\frac{C_{\star}N(s)}{n}.
\label{eq:global_positive_comparison}
\end{equation}
\end{mylemma}

\begin{myproof}
\textit{Regime-Stable Cells: global upper bound.}
We first study the uniform lim sup behaviour of the regime-stable cells contribution
\(
M_n^{(0)}(\theta\mid s^{\Delta_n})-M_\infty^{(0)}(\theta\mid s).
\)
For any \(s\in S_f([t_0,l])\), the number of regime-stable cells satisfies
\(
|G_n(s)|\le n.
\)
Since on each regime-stable cell the regime is constant, Proposition~\ref{prop:rdcm_uniform}
applies regime-wise. Because the corresponding bound is uniform on
\(
[t_0,l]\times\Theta^{\times2}
\)
and the number of regimes is finite, there exists a sequence of nonnegative random variables
\(
\{r_n\}_{n\ge1}
\)
with
\(
r_n=o_P(1),
\)
independent of \(s\), such that for every \(s\in S_f([t_0,l])\) and every \(n\ge1\),
\begin{equation}
\sup_{\theta\in\Theta^{\times2}}
\Big[
M_n^{(0)}(\theta\mid s^{\Delta_n})-M_\infty^{(0)}(\theta\mid s)
\Big]_+
\le
\frac{r_n|G_n(s)|}{n}
\le r_n.
\label{ehlafatica0}
\end{equation}

\medskip
\noindent
\textit{Switch Cells: global upper bound}. 
For the switch cells, we do not need a two-sided control of the whole local discrepancy. It is enough to bound from above the positive part of the regime-wise log-scaled bridge density
$
\Bigl[\log \tilde{f}_{t_i\mid t_{i-1},\,s^{\Delta_n}_{t_{i-1}}}(\theta)\Bigr]_+ ,
$
appearing in \eqref{eq:rescaled_bridge_dens_my0}.\footnote{We write \([\cdot]_+\) for the positive part operator.}
To prove this, it is convenient to express the local bridge density through the corresponding regime-wise bridge moments, namely
\[
k_{t_i\mid t_{i-1},\,s^{\Delta_n}_{t_{i-1}}}:=k_{t_i\mid t_{i-1},\,s^{\Delta_n}_{t_{i-1}}}
\bigl(X_{t_{i-1}};\theta_{s^{\Delta_n}_{t_{i-1}}}\bigr), \qquad \sigma^2_{t_i\mid t_{i-1},\,s^{\Delta_n}_{t_{i-1}}}:=\sigma^2_{t_i\mid t_{i-1},\,s^{\Delta_n}_{t_{i-1}}}
\bigl(\theta_{s^{\Delta_n}_{t_{i-1}}}\bigr)
\]
in \eqref{k}, and in \eqref{sigma} respectively.\footnote{For the bridge moments, we only display the arguments relevant for the present bound, namely the parameter block \(\theta\) and the stochastic input \(X_{t_{i-1}}\).} 
By Lemma~\ref{Factor}, the bridge variance
\(
\sigma^2_{t_i\mid t_{i-1},\,s^{\Delta_n}_{t_{i-1}}}
\)
factorizes as the product of
\(
v_{t_i\mid t_{i-1},\,s^{\Delta_n}_{t_{i-1}}}
\)
and the ratio
\(
R_{t_i\mid t_{i-1}}\in(0,1).
\)
Moreover, this ratio is uniformly bounded away from zero. Indeed, its numerator
\(
v_{T_{s^{\Delta_n}_{t_{i-1}}}\mid t_i}
\)
is bounded from below by a strictly positive constant, uniformly on
\(
[t_0,l]\times\Theta^{\times2},
\)
by the regime-wise ellipticity condition in Assumption~\ref{ass:rdcm_hf}.
\footnote{The denominator is the sum of the numerator and a nonnegative one-step variance term, and both quantities are uniformly controlled on the relevant space-time-parameter domain by continuity of the regime-wise diffusion coefficient and compactness, the ratio is uniformly bounded away from zero and above by \(1\).} Hence the rescaled bridge variance is uniformly controlled on
\(
[t_0,l]\times\Theta^{\times2}.
\)
The only stochastic term in the log-rescaled density is the quadratic Gaussian kernel, which enters with negative sign and is therefore bounded above by \(0\). It follows that there exists a deterministic constant \(C_2>0\), uniform on
\(
[t_0,l]\times\Theta^{\times2},
\)
such that
\begin{equation}
\Bigl[
\log \tilde{f}_{t_i\mid t_{i-1},\,s^{\Delta_n}_{t_{i-1}}}(\theta)
\Bigr]_+
\le C_2.
\label{Ehlafatica}
\end{equation}
Now we go directly on the upper limit bound of the switch cells. We have to consider the following set of inequalities that exploits the property $\left[\sum_{j}y_j\right]_{+}\leq \sum_{j}\left[y_j\right]_+$:
\begin{eqnarray}
M_n^{(J)}(\theta\mid s^{\Delta_n})&\leq&\left[M_n^{(J)}(\theta\mid s^{\Delta_n})\right]_{+}\leq \frac{1}{n}\underset{i \in J_n\left(s\right)}{\sum}
\left[\log
\tilde f_{t_i\mid t_{i-1},s^{\Delta_n}_{t_{i-1}}}\right]_+\nonumber\\
&\leq& \frac{C_2\left|J_{n}\left(s\right)\right|}{n}\leq \frac{C_2\left|N\left(s\right)\right|}{n}.\nonumber\\
\label{ehlafatica2}
\end{eqnarray}
the last inequality holds by Lemma~\ref{lem:one_jump_per_cell}, definitely we have
\(
|J_n(s)|\le N(s).
\)
Moreover, by point (ii) in Lemma~\ref{lem:pathwise_uniform_switch},
\[
\sup_{\theta\in\Theta^{\times2}}
\left|M_\infty^{(J)}(\theta\mid s)\right|
\leq
\frac{C_1 N(s)}{n}.
\]
Therefore,
\[
\Bigl[
M_n^{(J)}(\theta\mid s^{\Delta_n})
-
M_\infty^{(J)}(\theta\mid s)
\Bigr]_+
\leq
\Bigl[
M_n^{(J)}(\theta\mid s^{\Delta_n})
\Bigr]_+
+
\Bigl|
M_\infty^{(J)}(\theta\mid s)
\Bigr|,
\]
and hence
\begin{equation}
\sup_{\theta\in\Theta^{\times2}}
\Bigl[
M_n^{(J)}(\theta\mid s^{\Delta_n})
-
M_\infty^{(J)}(\theta\mid s)
\Bigr]_+
\leq
\frac{(C_1+C_2)N(s)}{n}.
\label{ehlafatica3}
\end{equation}
Setting
\(
C_\star:=C_1+C_2,
\)
and combining \eqref{ehlafatica0} with \eqref{ehlafatica3}, we obtain \eqref{eq:global_positive_comparison}.
\end{myproof}
We now upgrade the path-wise uniform control obtained in Lemma~\ref{lem:pathwise_uniform_switch}
to cylinder neighborhoods, which will be used in the lower-bound argument for the
log-sum-exp contrast. The key point is that, on a cylinder set, the number of switches
and the regime ordering are fixed, while the switch times vary only within compact disjoint
intervals.

Hence the path-wise estimates derived above extend uniformly over the whole
cylinder. The next lemmas should therefore be viewed as routine adaptations of the
path-wise control in Lemma~\ref{lem:pathwise_uniform_switch} to the present switching
setting, combining the standard switch-time viewpoint for finite-state continuous-time
Markov chains; see \cite{norris1997markov,andersonkurtz2015ctmc}, with a switch-time
approximation logic on finer and finer grids, in the spirit of
\cite{maller2008garch,iacus2018discrete}.

\begin{mylemma}
\label{lem:positive_cylinder_prob}
Every non-empty cylinder neighborhood $C\subset S_f([t_0,l])$ has strictly positive
$\Pi_Q$-probability.
\end{mylemma}

\begin{myproof}
A cylinder neighborhood fixes a finite sequence of regimes and constrains finitely many
switch times to lie in open intervals. Since all off-diagonal intensities are strictly
positive, the corresponding finite-dimensional switch-time density is strictly positive
on those open sets.
\end{myproof}

\begin{mylemma}
\label{lem:lipschitz_switch}
For every cylinder neighborhood $C\subset S_f([t_0,l])$ there exists $L_n(C)=O_p(1)$ such that
\[
\sup_{s\in C}
\sup_{\theta\neq\vartheta\in\Theta^{\times 2}}
\frac{|M_n(\theta\mid s^{\Delta_n})-M_n(\vartheta\mid s^{\Delta_n})|}
{\|\theta-\vartheta\|_2}
\le L_n(C).
\]
Moreover, there exists $L_\infty<\infty$ such that
\[
\sup_{s\in S_f([t_0,l])}
\sup_{\theta\neq\vartheta\in\Theta^{\times 2}}
\frac{|M_\infty(\theta\mid s)-M_\infty(\vartheta\mid s)|}
{\|\theta-\vartheta\|_2}
\le L_\infty.
\]
\end{mylemma}

\begin{myproof}
On each regime-stable cell, the one-step contrast is of RDCM type, and its dependence on the
relevant diffusion block is uniformly Lipschitz by the same argument used in the RDCM
proof. Summing over $O(n)$ regime-stable cells and using the prefactor $1/n$ gives an overall
$O_p(1)$ modulus. On switch cells there are only finitely many terms on a cylinder, and
each is continuous on the compact parameter set, so the switch contribution is $O_p(1/n)$.
The bound for $M_\infty$ follows immediately from \eqref{eq:Uinfty_path_switch_main}
and the uniform Lipschitz continuity of the integrands $\Psi_j$.
\end{myproof}

\begin{mylemma}
\label{lem:cylinder_uniform_switch}
For every cylinder neighborhood $C\subset S_f([t_0,l])$,
\[
\sup_{s\in C}
\sup_{\theta\in\Theta^{\times 2}}
|M_n(\theta\mid s^{\Delta_n})-M_\infty(\theta\mid s)|
\stackrel{\mathbb P}{\longrightarrow}0.
\]
\end{mylemma}

\begin{myproof}
All paths in a cylinder have the same number of switches and the same regime ordering, while
switch times vary in compact disjoint intervals. The RDCM expansion on the regime-stable cells is
uniform on the cylinder, and the switch-remainder is uniformly bounded by $K/n$ times an
$O_p(1)$ random factor, where $K$ is the common number of switches in the cylinder.
Therefore the same argument as in Lemma~\ref{lem:pathwise_uniform_switch} applies uniformly.
\end{myproof}

\begin{myproof}[\textbf{Proof of Theorem~\ref{thm:srdcm_uniform}}]\label{proof:srdcm_uniform}
We prove an upper and a lower bound.

\medskip
\noindent
\textit{Upper bound.}
By Lemma~\ref{lem:global_positive_comparison} 
% \footnote{Lemma~\ref{lem:pathwise_uniform_switch}
% already yields, for each fixed path $s\in S_f([t_0,l])$, a bound of the form
% $o_p(1)+O_p(1)N(s)/n$. Hence, on classes of paths with uniformly bounded jump counts,
% the path-wise control is already available. The role of
% Assumption~\ref{ass:srdcm_global_upper} is instead to prevent highly switching
% trajectories, that is, paths with large $N(s)$, from dominating the upper bound after
% integration in the observed log-sum-exp contrast. This is also consistent with the
% economic interpretation of the model, since frequent revisions of the perceived transition
% deadline would correspond to repeated costly adjustments in agents' decisions.}
, for every
$s\in S_f([t_0,l])$ and every $\theta\in\Theta^{\times 2}$,
\[
M_n(\theta\mid s^{\Delta_n})
\le
M_\infty(\theta\mid s)+r_n+\frac{C_\ast}{n}N(s),
\]
where the nonnegative quantity $r_n=o_{\mathbb P}(1)$ is independent of $(s,\theta)$.
Exponentiating and integrating with respect to $\Pi_Q$, we obtain
\[
\int_{S_f([t_0,l])}
\exp\!\bigl(nM_n(\theta\mid s^{\Delta_n})\bigr)\,\Pi_Q(ds)
\le
e^{nr_n}
\int_{S_f([t_0,l])}
\exp\!\bigl(nM_\infty(\theta\mid s)\bigr)e^{C_\ast N(s)}\,\Pi_Q(ds).
\]
Since
\[
M_\infty(\theta\mid s)\le M_\infty(\theta),
\qquad
\forall s\in S_f([t_0,l]),
\]
it follows that
\[
\int_{S_f([t_0,l])}
\exp\!\bigl(nM_n(\theta\mid s^{\Delta_n})\bigr)\,\Pi_Q(ds)
\le
e^{n(M_\infty(\theta)+r_n)}
E_{\Pi_Q}\!\bigl[e^{C_\ast N(S)}\bigr].
\]
Taking logarithms and dividing by $n$ gives
\[
M_n(\theta)-M_\infty(\theta)
\le
r_n+\frac1n\log E_{\Pi_Q}\!\bigl[e^{C_\ast N(S)}\bigr].
\]
Hence
\[
\sup_{\theta\in\Theta^{\times 2}}
\bigl(M_n(\theta)-M_\infty(\theta)\bigr)^+
\le
r_n+\frac1n\log E_{\Pi_Q}\!\bigl[e^{C_\ast N(S)}\bigr].
\]
By Assumption~\ref{ass:srdcm_hf}, the latent process is a finite-state continuous-time
Markov chain on the bounded interval $[t_0,l]$, hence $N(S)$ admits finite
exponential moments. Therefore
\[
\frac1n\log E_{\Pi_Q}\!\bigl[e^{C_\ast N(S)}\bigr]\to 0,
\]
and consequently
\[
\sup_{\theta\in\Theta^{\times 2}}
\bigl(M_n(\theta)-M_\infty(\theta)\bigr)^+
\xrightarrow{\mathbb P}0.
\]

\medskip
\noindent
\textit{Lower bound.}
Fix $\varepsilon>0$. By compactness of $\Theta^{\times 2}$, choose a finite
$\eta$-net
\[
\{\theta^{(1)},\dots,\theta^{(m)}\}\subset\Theta^{\times 2}.
\]
For each $k=1,\dots,m$, choose a cylinder neighborhood $C_k\subset S_f([t_0,l])$
such that
\[
\inf_{s\in C_k}M_\infty(\theta^{(k)}\mid s)\ge M_\infty(\theta^{(k)})-\varepsilon.
\]
By Lemma~\ref{lem:positive_cylinder_prob}, $\Pi_Q(C_k)>0$.

By Lemma~\ref{lem:cylinder_uniform_switch},
\[
\sup_{s\in C_k}\sup_{\theta\in\Theta^{\times 2}}
|M_n(\theta\mid s^{\Delta_n})-M_\infty(\theta\mid s)|
\xrightarrow{\mathbb P}0.
\]
Hence, for each fixed $k$,
\[
\inf_{s\in C_k}M_n(\theta^{(k)}\mid s^{\Delta_n})
\ge
M_\infty(\theta^{(k)})-2\varepsilon
\]
with probability tending to one.

Now fix $\theta\in\Theta^{\times 2}$ and let $k(\theta)\in\{1,\dots,m\}$ be such that $\|\theta-\theta^{(k(\theta))}\|_2\le\eta.$
% \[
% \|\theta-\theta^{(k(\theta))}\|_2\le\eta.
% \]
By Lemma~\ref{lem:lipschitz_switch}, for every $s\in C_{k(\theta)}$,
\[
M_n(\theta\mid s^{\Delta_n})
\ge
M_n(\theta^{(k(\theta))}\mid s^{\Delta_n})-L_n^\ast\eta,
\]
where $L_n^\ast:=\max_{1\le k\le m}L_n(C_k)=O_{\mathbb P}(1).$
% \[
% L_n^\ast:=\max_{1\le k\le m}L_n(C_k)=O_{\mathbb P}(1).
% \]
Therefore,
\[
\begin{aligned}
M_n(\theta)
&=
\frac1n\log
\int_{S_f([t_0,l])}
\exp\!\bigl(nM_n(\theta\mid s^{\Delta_n})\bigr)\,\Pi_Q(ds)
\\
&\ge
\frac1n\log
\int_{C_{k(\theta)}}
\exp\!\bigl(nM_n(\theta\mid s^{\Delta_n})\bigr)\,\Pi_Q(ds)
\\
&\ge
\frac1n\log
\int_{C_{k(\theta)}}
\exp\!\bigl(nM_n(\theta^{(k(\theta))}\mid s^{\Delta_n})-nL_n^\ast\eta\bigr)\,\Pi_Q(ds)
\\
&\ge
M_\infty(\theta^{(k(\theta))})-2\varepsilon-L_n^\ast\eta
+\frac1n\log \Pi_Q(C_{k(\theta)})
\end{aligned}
\]
with probability tending to one.

By Lemma~\ref{lem:lipschitz_switch}, the limit contrast is Lipschitz on
$\Theta^{\times 2}$, so
\[
M_\infty(\theta^{(k(\theta))})\ge M_\infty(\theta)-L_\infty\eta.
\]
Moreover, since the family $\{C_k\}_{k=1}^m$ is finite and each cylinder has strictly
positive $\Pi_Q$-probability,
\[
\min_{1\le k\le m}\Pi_Q(C_k)>0,
\qquad
\inf_{1\le k\le m}\frac1n\log \Pi_Q(C_k)=o(1).
\]
Thus
\[
M_n(\theta)\ge M_\infty(\theta)-2\varepsilon-(L_\infty+L_n^\ast)\eta+o(1)
\]
uniformly in $\theta$, with probability tending to one.

Choose first $\eta>0$ so small that $(L_\infty+A)\eta<\varepsilon$ on the event
$\{L_n^\ast\le A\}$, and then use that $L_n^\ast=O_{\mathbb P}(1)$. Since
$\varepsilon>0$ is arbitrary,
\[
\liminf_{n\to\infty}
\inf_{\theta\in\Theta^{\times 2}}
\bigl(M_n(\theta)-M_\infty(\theta)\bigr)\ge 0
\qquad\text{in probability.}
\]

Combining the upper and lower bounds yields
\[
\sup_{\theta\in\Theta^{\times 2}}
|M_n(\theta)-M_\infty(\theta)|
\xrightarrow{\mathbb P}0.
\]
\end{myproof}

\begin{myproof}[\textbf{Proof of Corollary~\ref{cor:srdcm_post_tau}}]
\label{proof:srdcm_post_tau}
By Theorem~\ref{thm:srdcm_uniform},
% \[
% \sup_{\theta\in\Theta^{\times 2}}
% |M_n(\theta)-M_\infty(\theta)|
% \stackrel{\mathbb P}{\longrightarrow}0.
% \]
under $l>\max\{\tau_1,\tau_2\}$, the whole diffusion block is visible in both regimes,
and the parameter space is compact. Hence the argmax theorem yields $\hat\theta_n\stackrel{\mathbb P}{\longrightarrow}\theta_0$.
\end{myproof}

\begin{myproof}[\textbf{Proof of Corollary~\ref{cor:srdcm_pre_tau}}]
\label{proof:srdcm_pre_tau}
If $l\le \tau_j$ for some regime $j$, then the regime-$j$ contribution to the limiting
contrast does not depend on $\theta_j^+$ and is therefore flat in that component. Hence
only the visible block $\theta_j^-$ is identifiable. Consistency of the corresponding
estimator follows the same arguments in Corollary \ref{cor:srdcm_post_tau} on the visible block.
\end{myproof}

\end{document}